\documentclass{article}
\usepackage{amsmath,amsthm, amssymb}
\usepackage{booktabs} 
\usepackage[ruled,linesnumbered]{algorithm2e} 

\usepackage{tikz}
\usetikzlibrary{positioning,shadows,shapes,backgrounds,fit}
\usepackage{natbib}
\usepackage{enumitem, tablefootnote}
\usepackage[paperwidth=8.5in,paperheight=11in,margin=1in]{geometry}
\usepackage{subcaption} 
\usepackage{comment}

\SetAlFnt{\small}
\SetAlCapFnt{\small}
\SetAlCapNameFnt{\small}
\SetAlCapHSkip{0pt}
\IncMargin{-\parindent}

\newtheorem{observation}{Observation}
\newtheorem{definition}{Definition}
\newtheorem{lemma}{Lemma}
\newtheorem{theorem}{Theorem}
\newtheorem{corollary}{Corollary}
\newtheorem{example}{Example}

\renewcommand{\b}{\mathbf{b}}
\newcommand{\M}{\mathcal{M}}
\newcommand{\T}{\mathcal{T}}
\newcommand{\vgraph}{OSP-graph}
\newcommand{\ver}{\mathcal{O}_i^\T}
\newcommand{\verp}{\mathcal{O}_i^{\T'}}

\newcommand{\vect}[1]{\mathbf{#1} }
\renewcommand{\a}{\vect a}

\newcommand{\bi}{\b_{-i}}
\renewcommand{\xi}{\vect{x}_{-i}}

\newcommand{\x}{\vect{x}}
\newcommand{\y}{\vect{y}}
\newcommand{\z}{\vect{z}}
\newcommand{\s}{\vect{s}}
\renewcommand{\t}{\vect{t}}

\renewcommand{\c}{\vect{c}}
\renewcommand{\d}{\vect{d}}

\newcommand{\bb}[1]{\vect{b}^{(#1)}}
\newcommand{\bbbb}[2]{b^{(#1)}_{#2}}
\newcommand{\bbb}[1]{b^{(#1)}_i}
\newcommand{\cc}[1]{\vect{c}^{(#1)}}
\newcommand{\ccc}[1]{c^{(#1)}_i}
\renewcommand{\a}{\vect a}
\newcommand{\ab}[1]{\a^{(#1)}}

\newcommand{\out}{\mathcal S}

\begin{document}


\title{On the Connection between Greedy Algorithms and Imperfect Rationality}

\author{Diodato Ferraioli\thanks{Universit\`a di Salerno, Italy. Email: {\tt dferraioli@unisa.it}} \and Carmine Ventre\thanks{King's College London, UK. Email: {\tt carmine.ventre@kcl.ac.uk}}}

\date{}

\maketitle


\begin{abstract}
The design of algorithms or protocols that are able to align the goals of the planner with the selfish interests of the agents involved in these protocols is of paramount importance in almost every decentralized setting (such as, computer networks, markets, etc.) as shown by the rich literature in Mechanism Design. 
Recently, huge interest has been devoted to the design of mechanisms for imperfectly rational agents, i.e., mechanisms for which agents are able to easily grasp that there is no action different from following the protocol that would satisfy their interests better. This work has culminated in the definition of Obviously Strategyproof (OSP) Mechanisms, that have been shown to capture the incentives of agents without contingent reasoning skills. 



Without an understanding of the algorithmic nature of OSP mechanisms, it is hard to assess how well these mechanisms can satisfy the goals of the planner. For the case of binary allocation problems and agents whose private type is a single number, recent work has shown that a generalization of greedy completely characterizes OSP. 
In this work, we strengthen the connection between greedy and OSP by providing 
a characterization of OSP mechanisms for all optimization problems involving these single-parameter agents. 
Specifically, we prove that OSP mechanisms must essentially work as follows: they either greedily look for agents with ``better'' types and allocate them larger outcomes; or reverse greedily look for agents with ``worse'' types and allocate them smaller outcomes; or, finally, split the domain of agents in ``good'' and ``bad'' types, and subsequently proceed in a reverse greedy fashion for the former and greedily for the latter. We further demonstrate how to use this characterization to give bounds on the approximation guarantee of OSP mechanisms for the well known scheduling related machines problem. 
\end{abstract}

\newpage

\setcounter{page}{1}


\section{Introduction}
\label{sec:Introduction}
The work in economics and computation 
has provided many protocols 
that align the goals of the planner with the selfish interests of agents involved in these protocols: examples range from network protocols (e.g., TCP \cite{akella2002selfish}, BGP \cite{nisan2011best}, Ethernet \cite{chen2007contention,feldman2016online, adamczyk2018random}) where this alignment was an indirect goal, to market protocols (e.g., ad auctions \cite{varian2007position}, spectrum auctions \cite{milgrom2004putting}) and the more recent blockchain protocols \cite{chiu2019incentive} for which the field has evolved to consider incentives more formally. 
The development of this field, 
known as \emph{mechanism design}, in fact recognizes that agents may, in principle, have an advantage if they deviate from the protocol's prescriptions. This could invalidate the guarantees of the protocol (such as, the maximization of  some social measure of welfare or the revenue of the designer) that only hold under the assumption that agents behave as dictated.
%
In Mechanism Design, the aim is to design special protocols, termed \emph{mechanisms}, that have both good performances and are compatible to the incentives of the agents, that is, it is in the agents' best interests to follow the protocol --- a property often termed \emph{strategy-proofness} (SP).

This line of research has led to the development of deep theoretical contributions and many 
fascinating results, that either design efficient and performing mechanisms in many settings \cite{book,archer2001truthful}, or prove their impossibility \cite{DSS15,CKK21}. Still, very few of the mechanisms designed in the literature have found practical applications. Arguably, one of the main reasons for this scarce applicability is the complexity of these mechanisms: not only they are too cumbersome to implement, but it is also too complex for agents to understand that there is no advantage in deviating from 
the protocol.

In order to address this issue,  the design of ``\emph{simple}'' mechanisms has attracted huge interest. In particular, we distinguish two approaches. Some work provides definitions of simplicity that turn out to be only partially satisfactory. For example, 
\citet{hartline2009simple} defined simplicity in a comparative way, so that one mechanism is more complex of another if the former can be simulated by the latter if we add few more agents; this definition is inspired by a seminal result by 
\citet{bulow1996auctions}, proving that a revenue-maximizing mechanism for selling a single item provides the same revenue achieved by a social-welfare maximizing mechanism with one more agent. While this definition is useful to compare different mechanisms, it does not necessarily lead to mechanisms that avoid the complexities described above.

A 
different approach is to propose mechanism that are ``trivially'' simple, the most prominent example being \emph{posted price} mechanisms. In these mechanisms, a price is associated a-priori to each possible action of the agents; hence, it is immediate for the agent to decide which action is the most convenient. Many posted price mechanisms have been proposed in a variety of settings \cite{babaioff2014simple,adamczyk2015sequential,feldman2017makespan,eden2017simple,correa2017posted}. Unfortunately, posted price mechanisms are also known to have severe limitations. For example, strong impossibility or inapproximability results for these mechanisms have been proved in diverse settings (cf., \cite{adamczyk2015sequential,babaioff2017posting}). Clearly, without a precise definition of simplicity, it is impossible to understand if these impossibility results hold only for this specific mechanism format or for every simple mechanism.

A specific notion has emerged in the literature in economics, with the definition of 
\emph{Obviously Strategyproof (OSP)} mechanisms
\cite{liosp}. These are proved to match the concept of simplicity we seek in that even agents with imperfect rationality, namely those lacking contingent reasoning skills, understand that it is best for them to play according to the mechanism's rules. Roughly speaking, a mechanism is OSP if whenever it requires an agent to take an action, the worst outcome that she can achieve by following the protocol is not worse than the best outcome that she can achieve by deviating. It is evident that posted price mechanisms are OSP, but it is not hard to see that other mechanism formats satisfy this property. Consider, for example, an English auction (a.k.a., ascending price auction) where the price to sell a single item is raised at each time step, and agents have to decide whether to leave or remain in the mechanism. The last agent left in the auction is the only winner -- she receives the item at the current price. Following the mechanism in this context means that each agent participates to the mechanism as long as the price is below her valuation for the item, and leaves when the price is too high. The best case when deviating from this rule and leaving the mechanism when the price is below one's valuation leads to 
losing the item, which is certainly not better than the worst outcome when following the protocol. Similarly, staying in the auction when the price is above the valuation can only lead to outcomes for which the agent has a non-positive utility (losing the item or winning it for a price above her valuation) which is worse than following the protocol (which always guarantees non-negative utility); hence, deviating is again not better than following the protocol.

OSP mechanisms have attracted a lot of attention in both economics and computer science. Some works provide preliminary characterizations for these mechanisms akin the revelation principle for SP (which does not hold for OSP) 
\cite{badegonczarowski,pycia2,mackenzie}.
This research allows to think, without loss of generality, at deterministic (rather than randomized) extensive-form mechanisms where each agent moves sequentially (rather than concurrently). More relevant to our paper is a technique to characterize OSP via \emph{cycle-monotonicity}, which has been defined by \citet{MOR22}.

Fewer results are instead known on the construction of these mechanisms. Most of these results focus on restricted preferences, such as single-peaked domains \cite{badegonczarowski,AMN19,AMN20} whereas others focus on specific applications, e.g., stable matching  \cite{ashlagigonczarowski}, machine scheduling \cite{MOR22} and binary allocation problems \cite{MOR22}. Negative results, such as inapproximability or impossibility results, about the performances of OSP mechanisms are similarly quite sparse. Some inapproximability results have been instead provided for special mechanisms formats, that can be observed to be OSP, such as \emph{deferred acceptance} auctions \cite{MS20}. For example, \citet{DGR17} prove that the approximation guarantee of these mechanisms are quite poor compared to what strategyproof mechanisms can do for several optimization problems, as confirmed by more recent work along this line \cite{CGS22,FGGS22}. 
It is unknown whether these results extend to any OSP mechanism.

\paragraph{Our Contribution.}
We focus on the case of selfish agents with a type space that consists of the set of real numbers, i.e.,  \emph{single-parameter} agents. Despite its simplicity, this setup still allows to model many fundamental optimization problems, see, e.g., \cite{book,AT01}. Moreover, the extent to which is possible to incentivate these 
agents has been widely studied and is well understood for perfectly rational agents. However, 
we are very far from understanding how easy it is to design OSP mechanisms 
or establish their limits even for this simple setup. 

Recently, \citet{fpv21} provided a characterization of OSP mechanisms for single-parameter problems \emph{with binary outcomes}, i.e., where agents are either selected or not in the eventual solution. They show that a mechanism is OSP if and only if it employs a \emph{two-way greedy} algorithm. 
This is in essence either a greedy algorithm (i.e., selecting agents with ``good'' types as long as it is feasible) or reverse greedy (a.k.a., deferred acceptance) algorithms (i.e.,  discarding agents with ``bad'' types), with the possibility of interleaving the two approaches only in some rare 
cases. Let us focus on the Minimum Spanning Tree (MST) problem to make the difference between forward and reverse greedy more explicit. Agents here control the edges of a graph and their type is the cost for using the edge. 
Greedy algorithms incrementally build the MST by selecting the agent with the lowest cost (a ``good'' type to use the terminology above) that do not close a cycle, whereas reverse greedy rejects agents with the highest costs until it is left with a spanning tree. The great advantage of this characterization is that 
one can design OSP mechanisms or evaluate their limits, by simply importing well-known results about greedy and reverse greedy algorithms, cf. bounds in \cite{fpv21}.

The techniques of \citet{fpv21} seem tailored to the binary outcome setting, and therefore it is unclear to what extent the connection between greedy algorithms and OSP mechanisms holds more generally. 
In this work, we provide such a result: essentially, we show that every algorithm for single-parameter optimization algorithms can be turned into an OSP mechanism if and only if it is \emph{three-way greedy}, i.e., it has the following format: either it is a greedy algorithm or a reverse greedy algorithm or a carefully built combination of the two. In this context, a greedy algorithm looks for agents with ``good'' types (e.g., low costs) and allocates them outcomes that are monotone in their type (e.g., non-increasing in the cost). Similarly, a reverse greedy algorithm  looks for agents with ``bad'' types (e.g., high costs) and allocates them monotone outcomes (e.g., non-increasing in the cost). The combination of the two approaches instead defines two sets: one containing ``good'' types (e.g., low costs) and the other comprised of ``bad'' types (e.g., high costs) so that each type in the good set is better than the types in the bad set. The algorithm guarantees that the outcomes allocated to the types in the good set are larger than the outcomes assigned to those in the bad set. At this point, the algorithm runs a reverse greedy algorithm on the  set of good types, and a greedy algorithm on the  set of bad types.

As in the binary outcome case, we have that interleaving between greedy and reverse greedy may actually occur. We provide a characterization of the cases in which this interleaving can occur: intuitively, this is either when the outcome of the agent is essentially revealed (i.e., it does not depend on the actions of other agents actions), as in the case of two outcomes; or when interleaving involves types that are clearly distinguishable from the rest of the domain, meaning that either they are associated to outcomes that are far away from the ones allocated to types in the rest of the domain, or the value of these types is sufficiently far away from the rest of the domain.

From a technical point of view, our characterization is based on the cycle monotonicity characterization provided by \citet{MOR22}. However, in order to work with negative cycles of arbitrary length, we rely on a contribution that 
could be of independent interest. Namely, we introduce some \emph{ironing} steps 
that lead to cycles with a ``canonical'' structure. Leveraging this structure, we then show that we can restrict 
without loss of generality to OSP mechanisms whose implementation is \emph{outcome-monotone}: i.e., queries to the agents are feasible only when all the possible outcomes resulting from taking one action are smaller than the possible outcomes resulting from taking an alternative action. 

We apply our results to show a lower bound on the approximation ratio of {any} OSP mechanism for the well studied problem of scheduling on selfish related machines. Agents have as private type (the inverse of) the speed of the machine they own and we are interested in minimizing the makespan (i.e., the latest completion time of a machine). It was known that the bound is in the interval $(\sqrt{n}, n]$ for $n$ selfish machines. In particular, the lower bound was proved by looking at cycles of length two whilst the upper bound can be simply obtained through an ascending (or descending) mechanism to find the speed of each machine \cite{MOR22}. Interestingly, $\sqrt{n}$ is known to be tight for domains of size three whereas the optimum is possible with two types only  \cite{MOR22}. Thanks to our characterization, we are able to show that with five or more types, each mechanism with approximation lower than $n$ needs to necessarily create a cycle of length four that has negative weight, and thus it is not OSP. Moreover, when the agent domains have four types only, we prove a lower bound of $n/2+1$ and give some indication that this is tight.
This contribution not only closes the gap left open in \cite{MOR22} but also shows (i) how to work with the seemingly unwieldy characterization of OSP mechanisms for general optimization problems; and, (ii) that to some extent, it is sufficient to focus on cycles of length four to obtain bounds on the approximation ratio of OSP mechanisms.

\section{Preliminaries and Notation}
We let $N$ denote a set of $n$ \emph{selfish agents} and $\out$ a set of feasible \emph{outcomes}. Each agent $i$ has a \emph{type} $t_i \in D_i$ that we assume to be her \emph{private knowledge}. We call $D_i$ the \emph{domain} of $i$. With $t_i(X) \in \mathbb{R}$ we denote the \emph{cost} of agent $i$ with type $t_i$ for the outcome $X \in \out$. When costs are negative, the agent has a {profit} from the solution, called \emph{valuation}. We will be working with costs and use that terminology accordingly but our results do not assume that costs are positive. 

A \emph{mechanism} interacts with the agents in $N$ to select an outcome $X \in \out$. Specifically, agent $i$ takes \emph{actions} (e.g., saying yes/no) that may signal to the mechanism a type $b_i \in D_i$ different from $t_i$ (e.g., saying yes could signal that the type has some properties that $b_i$ has but $t_i$ does not). 
We then say that agent $i$ takes \emph{actions compatible with (or according to) $b_i$} and call $b_i$ the presumed type. 

For a mechanism $\M$, $\M(\b)$ denotes the outcome returned by the mechanism when the agents take actions according to their presumed types $\b = (b_1, \ldots, b_n)$  (i.e., each agent $i$ takes actions compatible with the corresponding $b_i$).
This outcome is computed by a pair $(f,p)$, where $f = f(\b)=(f_1(\b),\ldots,f_n(\b))$ (termed \emph{social choice function} or algorithm) maps the actions taken by the agents according to $\b$ to a feasible solution in $\out$,
and $p(\b)=(p_1(\b),\ldots,p_n(\b)) \in \mathbb{R}^n$ maps the actions taken by the agents according to $\b$ to \emph{payments}. 

Each selfish agent is equipped with a \emph{quasi-linear utility function}, i.e., agent $i$ has utility function $u_i \colon D_i \times \out \rightarrow \mathbb{R}$: for $t_i \in D_i$ and for an outcome $X \in \out$ returned by a mechanism $\M$, $u_i(t_i, X)$ is the utility that agent $i$ has for the implementation of outcome $X$ when her type is $t_i$, i.e., 
$
u_i(t_i, \M(b_i, \bi)) = p_i(b_i, \bi) - t_i(f(b_i, \bi)).
$

A \emph{single-parameter} agent $i$ has as private information  a single real number $t_i$ and $t_i(X)$ can be expressed as $t_i \mathsf{w}_i(X)$ for some publicly known function $\mathsf{w}$; note that $\mathsf{w}_i(X)$ is a non-negative real number (and $\out=\mathbb{R}^n_{\geq 0}$).  Moreover, observe that the cost of player $i$ is independent on what the outcome $X$ prescribes for players different from $i$. We make no other assumption on $\out$. To simplify the notation, we will write $t_i f_i(\b)$ when we want to express the cost of a single-parameter agent $i$ of type $t_i$ for the output of social choice function $f$ on input the actions corresponding to a bid vector $\b$. 

\paragraph{Extensive-form Mechanisms and Obvious Strategyproofness.}
We here follow \citet{fpv21} and introduce the concept of implementation tree to formally define (deterministic) OSP mechanisms. As in prior work, our definition is built on the one by \citet{mackenzie} rather than the original definition by \citet{liosp}; see \cite{fpv21} for details.

An \emph{extensive-form mechanism} 
$\M$ is a triple $(f,p,\T)$ where, as from above, the pair $(f,p)$ determines the outcome of the mechanism, and $\T$ is a 
tree, called \emph{implementation tree}, such that:
\begin{itemize}[leftmargin=0.45cm, noitemsep, topsep=0pt]
	\item Every leaf $\ell$ of the tree is labeled with a possible outcome of the mechanism $(X(\ell), p(\ell))$, where $X(\ell) \in \out$ and $p(\ell) \in \mathbb{R}$;
	\item Each node $u$ in the implementation tree $\T$ defines the following:
    \begin{itemize}[leftmargin=0.6cm, noitemsep, topsep=0pt]
        \item An agent $i=i(u)$ to whom the mechanism makes a query. Each possible answer to this query leads to a different child of $u$.
        \item A subdomain $D^{(u)}=(D_i^{(u)}, D_{-i}^{(u)})$ containing all types that are \emph{compatible} with $u$, i.e., compatible with all the answers to the queries from the root down to node $u$.
        Specifically, the query at node $u$ defines a partition of the current domain of $i=i(u)$, $D_i^{(u)}$ into $k\geq 2$ subdomains, one for each of the $k$ children of node $u$. Thus, the domain of each of these children will have as the domain of $i$, the subdomain of $D_i^{(u)}$ corresponding to a different answer of $i$ at $u$, and an unchanged domain for the other agents.
\end{itemize}
\end{itemize}

Observe that, according to the definition above, for every profile $\b$ there is only one leaf $\ell = \ell(\b)$
such that $\b$ belongs to $D^{(\ell)}$.
Similarly, to each leaf $\ell$ there is at least a profile $\b$ that belongs to $D^{(\ell)}$.
For this reason, we say that $\M(\b) = (X(\ell), p(\ell))$.

Two profiles $\b$, $\b'$ are said to \emph{diverge} at a node $u$ of $\T$ if this node has two children $v, v'$ such that $\b \in D^{(v)}$, whereas $\b' \in D^{(v')}$. For every such node $u$, 
we say that 
$i(u)$ is the \emph{divergent agent} at $u$.

We are now ready to define obvious strategyproofness.
An extensive-form mechanism $\M$ is \emph{obviously strategy-proof (OSP)} if for every agent $i$ with real type $t_i$,
for every vertex $u$ such that $i = i(u)$, for every $\bi, \bi'$ (with $\bi'$ not necessarily different from $\bi$),
and for every $b_i \in D_i$, with $b_i \neq t_i$,
such that $(t_i, \bi)$ and $(b_i, \bi')$ are compatible with $u$, but diverge at $u$,
it holds that $u_i(t_i, \M(t_i, \bi)) \geq u_i(t_i,\M(b_i, \bi'))$.
Roughly speaking, an OSP 
mechanism requires that, at each time step
agent $i$ is asked to take a decision that depends on her type, the worst utility that
she can get if she behaves according to her true type
is at least the best utility she can get by behaving differently.
{We stress that our definition does not restrict the alternative behavior to be consistent with a fixed type. Indeed, as noted above, each leaf of the tree rooted in $u$, denoted $\T_u$, corresponds to a profile $\b = (b_i, \bi')$ compatible with $u$: then, our definition implies that the utility of $i$ in the leaves where she plays truthfully is at least as much as the utility in every other leaf of $\T_u$.}

\paragraph{Cycle-monotonicity Characterizes OSP Mechanisms.}
We next describe the main tools introduced by \citet{MOR22} to show that  OSP can be characterized by the absence of negative-weight cycles in a suitable weighted graph over the possible strategy profiles. 
We consider a mechanism $\M$ with implementation tree $\T$ for a social choice function $f$, and define: 
\begin{itemize}[leftmargin=0.45cm, noitemsep, topsep=0pt]
	\item \textbf{Separating Node:} A node $u$ in the implementation tree $\T$ is  $(\a,\b)$-separating for agent $i=i(u)$ if $\a$ and $\b$ are compatible with $u$ (that is, $\a,\b\in D^{(u)}$), and the two types $a_i$ and  $b_i$ belong to two different subdomains of the children of $u$ (thus implying $a_i\neq b_i$).
	\item \textbf{OSP-graph:} For every agent $i$, we define a 
	 directed weighted graph $\ver$ having a node for each profile in $D=\times_i D_i$. The graph 
	contains 
	edge $(\a, \b)$ if and only if $\T$ has some node $u$ which is $(\a,\b)$-separating for $i=i(u)$, and the weight of this edge is $w(\a, \b)= a_i(f_i(\b)- f_i(\a))$. Throughout the paper, we will denote with $\a \rightarrow \b$ an edge $(\a,\b) \in \ver$, and with $\a \rightsquigarrow \b$ a path among these two profiles in $\ver$.
	\item \textbf{OSP Cycle Monotonicity (OSP CMON):} 
	OSP cycle monotonicity (OSP CMON)  holds  if, for all $i$, the graph $\ver$ does not contain negative-weight cycles. Moreover, 
	OSP two-cycle monotonicity (OSP 2CMON) holds if the same is true when considering cycles of length two only, i.e., cycles with only two edges. Sometimes, we will simply say CMON and 2CMON below.
\end{itemize}

\begin{theorem}[\cite{MOR22,fpv21}]\label{thm:cmon}
	A mechanism with implementation tree $\T$ for a social function $f$ is OSP on finite domains if and only if  OSP CMON holds. Moreover, for any OSP mechanism $\M=(f, p, \T)$ where $\T$ is not a binary tree, there is an OSP mechanism $\M'=(f, p, \T')$ where $\T'$ is a binary tree.
\end{theorem}
Given the result above, we henceforth assume that the agents have finite domains and that  
the implementation trees of our mechanisms are binary. 

\paragraph{Scheduling Related Machines.}
We have a set of $m$ jobs to execute and the $n$ agents control related machines.
Agent $i$ has type $t_i$, a job-independent processing time per unit of job
(i.e., an execution speed $1/t_i$ for all the jobs). The set of feasible solutions is an allocation of the $m$ jobs to the $n$ machines. We let $\b$, with $b_i \in D_i$, denote a generic instance of the machine processing times. The social choice function $f$ {maps its input $\b$ to} 
a possible allocation $f(\b) = (f_1(\b), \ldots, f_n(\b))${,} 
where $f_i(\b)$ denotes the job load assigned to machine $i$. 
The cost that agent $i$ faces for the schedule $f(\b)$ {when her type is $b_i$} is $b_i(f(\b))=b_i \cdot f_i(\b)$. We focus on social choice functions $f^*$ {returning a schedule that minimizes} the \emph{makespan} {for every instance $\b$}, i.e., 
$
f^*(\b) \in \arg\min_{\x {\in \out}} {MS(\x,\b),}
$ 
{where {$MS(f,\b)=\max_{i=1}^n b_i(f_i(\b))$ denotes the makespan of solution $f(\b)$} according to machine processing times $\b$.} We say that $f$ is $\rho$-approximate if it returns a solution whose cost is at most $\rho$ times the optimum{, that is, for all instances $\b$, we have
	{$MS(f(\b),\b) \leq \rho \cdot MS(f^*(\b),\b).$}}

\section{The Characterization}
This section contains our characterization of OSP mechanisms. 
We first state a couple of technical ingredients needed for our proof: 
we show 
that any mechanism that is not CMON exhibits a certain antimonotone behaviour; then we define 
a necessary condition for OSP that does a lot of the leg work for the characterization (we defer the technical work to prove this important theorem to Section \ref{sec:ordered}).
We finally 
introduce the mechanism format and the proof of the characterization. 

\paragraph{Antimonotone Types.}
We begin with a structural property of negative-weight cycles. 

\begin{theorem}\label{thm:anatomy}
	Let $\M$ be a mechanism with implementation tree $\T$ and social choice function $f$ that is  OSP 2CMON but not OSP CMON. Then, for every negative-weight cycle $C$ in some OSP graph $\ver$ there are fours profiles, $\bb 0, \bb 1, \bb 2, \bb 3$, belonging to $C$ such that:
		(i) $\bbb 0 > \bbb 1 > \bbb 2 > \bbb 3$; 
		(ii) $f_i(\bb 1) > f_i(\bb 2)$; and,
		(iii) there is an edge between $\bb 1$ and $\bb 0$, an edge between $\bb 2$ and $\bb 3$ but no edge between $\bb 1$ and $\bb 2$ in $\ver$.
\end{theorem}
\begin{proof}
	Let $C$ be a negative-weight cycle in some \vgraph\ $\ver$. Since $C$ has negative-weight then there must be an edge $(\cc{neg1}, \cc{neg2})$ therein of negative-weight. This means that $k_1=f_i(\cc{neg1})>f_i(\cc{neg2})=k_2$ and, by OSP 2CMON, that $\ccc{neg1} < \ccc{neg2}$. Let $(\cc{pos1}, \cc{pos2})$ be the \emph{first} positive-weight edge following $(\cc{neg1}, \cc{neq2})$ in $C$ so that, by definition of edge weight and OSP 2CMON, $\ccc{pos1} > \ccc{pos2}$. 
	Note that such an edge must exist since the cycle will need to go from a profile for which $i$ has outcome $k_2$  (i.e., $\cc{neg2}$) to one where $f_i$ is $k_1$. For $C$ to be negative then it must be the case that the weight of 
	\begin{equation}\label{eq:negativepart}
		\cc{neg1} \rightarrow \cc{neg2}  \rightsquigarrow \cc{pos1} \rightarrow \cc{pos2}
	\end{equation}
	is negative. (If we cannot find no such four profiles in $C$ then for each negative-weight edge in $C$ there is a sequence of cycle edges following it that weigh more, a contraction.) 
	
	Let $\mathcal{P}$ denote the path 
	\[
	\cc{1}=\cc{neg1} \rightarrow \cc{2}=\cc{neg2}  \rightarrow \cc{3} \rightarrow \cdots \rightarrow \cc{\ell-1} \rightarrow \cc{\ell} = \cc{pos1}
	\]
	that is, the path between $\cc{neg1}$ and $\cc{pos1}$ (extremes included). Note that for convenience we are naming profiles in $C$ (including $\cc{neg1}, \cc{neg2}$ and $\cc{pos1}$) with consecutive indices. Similarly, we let $\cc{\ell+1} = \cc{pos2}$. Below, we also let $k_j$ denote a shorthand for $f_i(\cc{j})$.
	%
	%
	
	By construction, there cannot be positive-weight edges in $\mathcal{P}$ implying that $k_{j+1} \leq k_j < k_1$ for all $(\cc{j}, \cc{j+1}) \in \mathcal{P}$. We let $\mathcal{P}_{<0}=\left\{\cc{j} | k_{j+1}<k_j \right\}$, that is, all the profiles $\cc{j}$ whose outgoing edge in $\mathcal{P}$ has negative weight. By  \eqref{eq:negativepart} we then have
	\begin{align*}
		- (k_1-k_{\ell}) \ccc{\max} + (k_{\ell+1}-k_{\ell}) \ccc{pos1} & <\sum_{\cc{j} \in \mathcal{P}_{<0}} \ccc{j} \left(k_{j+1}-k_{j}\right) + (k_{\ell+1}-k_{\ell}) \ccc{pos1} \\ 
		& =  \sum_{j=1}^{\ell-1} \ccc{j} \left(k_{j+1}-k_{j}\right) + (k_{\ell+1}-k_{\ell}) \ccc{pos1} < 0,
	\end{align*}
	where $\cc{\max}$ is the profile in $\mathcal{P}_{<0}$ with maximum value of $\ccc{\max}$ (clearly, $\cc{pos1} \not\in \mathcal{P}_{<0}$). 
	%
	We now consider two cases depending upon the value of $f_i(\cc{pos2})$; recall that $k_{\ell+1}=f_i(\cc{pos2}) > f_i(\cc{pos1})=k_\ell$ since the weight of $(\cc{pos1}, \cc{pos2})$ is positive and that by construction $k_1 > k_\ell$. 
	
	If $f_i(\cc{pos2}) \geq f_i(\cc{neg1})$ (that is, $k_{\ell+1} \geq k_1$) then the proof concludes by setting $\bb{1} = \cc{\max}$, $\bb{0}$ to the profile following $\bb{1}$ in $\mathcal{P}$ (by definition of $\cc{\max}$, we have $\bbb{0}>\bbb{1}$), $\bb{2}=\cc{pos1}$ and $\bb{3}=\cc{pos2}$. In fact, it is the case that for some $1\leq j \leq \ell-1$, $f_i(\bb{1})=f_i(\cc{\max}) = k_j > k_{j+1} \geq f_i(\cc{pos1}) = f_i(\bb{2})$. Recall that by OSP 2CMON there is no edge indeed between $\bb{1}$ and $\bb{2}$.
		
	Assume then that $k_{\ell+1} < k_1$. We now let $(\cc{pos3}, \cc{pos4})$ denote the subsequent positive-weight edge in $C$; again this must exist since $k_{\ell+1}<k_1$. We are going to consider the weight of the path 
	\begin{equation}\label{eq:negativepart:2}
		\cc{neg1} \rightarrow \cc{neg2}  \rightsquigarrow \cc{pos1} \rightarrow \cc{pos2}  \rightsquigarrow \cc{pos3} \rightarrow \cc{pos4}
	\end{equation}
	in $C$. If this were non-negative then the rest of $C$ must have negative-weight and we can repeat the argument there (specifically, it is not possible that in the remainder of $C$ all  paths like \eqref{eq:negativepart} and \eqref{eq:negativepart:2} have non-negative weight). We can therefore assume that \eqref{eq:negativepart:2} is negative. 
	
	Similarly to above, we let $\mathcal{P}'$ denote the path in $C$ between $\cc{pos2}$ and $\cc{pos3}$ (and rename profiles) as follows
	$$
	\cc{pos2} = \cc{\ell + 1} \rightarrow \cc{\ell+2}  \rightarrow \cdots \rightarrow \cc{\kappa-1} \rightarrow \cc{\kappa} = \cc{pos3}.
	$$
	Following this formalism, we let $\cc{\kappa+1}=\cc{pos4}$; as above, $k_j$ is a shorthand for $f_i(\cc{j})$ for $\cc{j}$ in $\mathcal{P'}$ and $\mathcal{P}'_{<0}=\left\{\cc{j} | \ell+1 \leq j \leq \kappa \wedge w(\cc{j},\cc{j+1}) < 0\right\}$. We again have $k_{j+1} \leq k_j$, for $j \geq \ell+1$. Then 
	\begin{equation}\label{eq:twopos}
		\begin{split}
			& - (k_1-k_{\ell}) \ccc{\max} - (k_{\ell+1}-k_{\kappa}) \ccc{\max'} + (k_{\ell+1}-k_{\ell}) \ccc{pos1} + (k_{\kappa+1}-k_{\kappa}) \ccc{pos3} < \\ & \sum_{\cc{j} \in \mathcal{P}_{<0} \cup \mathcal{P}'_{<0}} \ccc{j} \left(k_{j+1}-k_{j}\right)  + (k_{\ell+1}-k_{\ell}) \ccc{pos1}+ (k_{\kappa+1}-k_{\kappa}) \ccc{pos3} =\\ 
			& \sum_{j=1, j \neq \ell}^{\kappa-1} \ccc{j} \left(k_{j+1} - k_{j}\right) + (k_{\ell+1}-k_{\ell}) \ccc{pos1}+ (k_{\kappa+1}-k_{\kappa}) \ccc{pos3} < 0,
		\end{split}
	\end{equation}
	where $\cc{\max'}$ is the profile in $\mathcal{P}'_{<0}$ with maximum value of $\ccc{\max'}$ (or a dummy profile with $\ccc{\max'}=0$ if $\mathcal{P}'_{<0}=\emptyset$). 
	Firstly note that, as above, $f_i(\cc{\max}) = k_j > k_{j+1} \geq f_i(\cc{pos1})$ for some $j < \ell$. Then if $\ccc{\max}$ were to be larger than $\ccc{pos1}$ then we could conclude the proof with the same definitions of $\bb{0}$, $\bb{1}$, $\bb{2}$ and $\bb{3}$ as in the case of the path \eqref{eq:negativepart}. Similarly, we have $f_i(\cc{\max'}) = k_j > k_{j+1} \geq f_i(\cc{pos3})$ for some $\kappa> j > \ell$ and, again, should it be that $\ccc{\max'}>\ccc{pos3}$ then we could similarly conclude the proof by choosing suitable profiles in $\mathcal{P'}$. We are 
	left with 
	\begin{equation}\label{eq:mixcase}
	\ccc{\max}\leq \ccc{pos1} \wedge \ccc{\max'}\leq \ccc{pos3}.
	\end{equation}
	By leveraging the signs of the factors of the types in \eqref{eq:twopos} and using \eqref{eq:mixcase} above, we obtain: 
	{\allowdisplaybreaks
	\begin{align*}
		- (k_1-k_{\ell+1}) \ccc{\max} + (k_{\kappa+1}-k_{\ell+1}) \ccc{pos3} & < 0,\\
		- (k_{\ell+1}-k_{\kappa+1}) \ccc{\max'} + (k_{\ell+1}-k_1) \ccc{pos1} & < 0.
	\end{align*}
}
	If $k_{\kappa+1} \geq k_1$, by \eqref{eq:mixcase}, we then get the contradiction that 
	\begin{align*}
		\ccc{\max} > \ccc{pos3} \geq \ccc{\max'} > \ccc{pos1} \geq \ccc{\max}.
	\end{align*}
	
	
	
	If instead $k_{\kappa+1} < k_1$ then we can repeat the argument and add the positive-weight edge following $(\cc{pos3}, \cc{pos4})$ in $C$. Since the cycle needs to go back to a node with outcome at least $k_1$ this process will terminate and will identify the needed pair of types with antimonotone outcomes (which, by OSP 2CMON, are not connected by an edge) with their neighbors in $C$.
\end{proof}

\begin{definition}[Antimonotone Types and Witness Profiles]
	We say that two types like $\bbb{1}$ and $\bbb{2}$ in the statement of the theorem above are \emph{antimonotone}  and call the profiles $\bb{1}$ and $\bb{2}$ \emph{witnesses} of antimonotonicity of $\bbb{1}$ and $\bbb{2}$.
\end{definition}

\paragraph{A Necessary Condition.}

Let us start by providing a useful definition.
\renewcommand{\L}{\mathcal{L}}

\begin{definition}[Monotone Labels of Types/Agents]
	Let $\M=(f,p, \T)$ be an extensive-form mechanism. For a generic $t \in D_i^{(u)}$, $u$ being a node of the implementation tree $\T$, we define the set of outcomes that $i$ can be allocated for $t$ in the subtree rooted at $u$, also called \emph{label} of $t$ for $i$ at $u$, as
	$
	\L_i^{(u)}(t)=\{f_i(t, \bi) \; | \; \bi \in D_{-i}^{(u)}\}.
	$ 
	Moreover, given two types $t, t' \in D_i^{(u)}$, $t> t'$, we let
	$
	\L_i^{(u)}(t) \preceq \L_i^{(u)}(t')$ 
 denote the case in which $x \leq y 
	\; \forall x \in \L_i^{(u)}(t)$ and 
	$y \in \L_i^{(u)}(t').
	$  
	If this is true for all types in the current domain at $u$, we say that the labels of $i$ are \emph{monotone} at $u$.
\end{definition} 
In words, a label of a type at a certain node $u$ of the implementation tree contains all the outcomes that the divergent agent can be allocated in the subtree $\T_u$ rooted at $u$. Note that if the labels at a certain $u$ are monotone for two types, then these cannot be antimonotone. If the labels at $u$ are monotone for agent $i$, it means that she has no antimonotone types in $D_i^{(u)}$.

%

The next notion introduces a particular query that not only does not break 2CMON but also keeps the types in the two parts ordered. Recall that 
we can 
restrict to binary implementation trees. 

\begin{definition}[Ordered Query]
	Let $\M=(f,p, \T)$ be an extensive-form mechanism. Let $u \in \T$ be a node where $i=i(u)$ and $D_i^{(u)}$ is separated into $L$ and $R$. We say that the query at $u$ is \emph{ordered} if for all $l,r \in D_i^{(u)}$ with $l$ in $L$ and $r$ in $R$, $l < r$ and $\L_i^{(u)}(r) \preceq \L_i^{(u)}(l) $.
\end{definition}

We say that a mechanism is ordered if it only makes ordered queries. Next result shows that we can focus on ordered mechanisms without loss of generality as long as we are interested in OSP. 
\begin{theorem}\label{thm:monowlog}
		Any OSP mechanism $\M=(f, p, \T)$ can be transformed into an equivalent OSP mechanism $\M'=(f, p, \T')$ where all queries in $\T'$ are ordered.
\end{theorem} 

\paragraph{The Mechanism Format.}
We begin by defining the concept of pivots for a pair of types. Intuitively, $\bbb{u}$ and $\bbb{d}$ are  pivots for $\bbb{1}$ and $\bbb{2}$ at $u \in \T$ if they have been separated from these types before $u$, and they have labels that are respectively at least as large as the larger label of $\bbb{1}$ and $\bbb{2}$ at $u$, and at least as small as the smaller label of $\bbb{1}$ and $\bbb{2}$ at $u$.

\begin{definition}[Pivots]\label{def:pivot}
Given a node $u$ and a pair of types $\bbb{1}, \bbb{2} \in D_i^{(u)}$,  
we say that types $\bbb{u}$, $\bbb{d}$ are \emph{pivots} for $\bbb{1}$ and $\bbb{2}$ if 
\begin{itemize}[topsep=0pt,leftmargin=0.45cm]
	\item they are separated from $\bbb{1}, \bbb{2}$ respectively at nodes $v^{(u)}, v^{(d)} \in \T$ that are ancestors of $u$ with $i(v^{(u)})=i(v^{(d)})=i$;
	\item for $y \in \L_i^{(u)}(\bbb{1})$ and $x \in \L_i^{(u)}(\bbb{2})$ with $y>x$, there are $z \in \L_i^{(v^{(u)})}(\bbb{u})$ with $z \geq y$, and $q \in \L_i^{(v^{(d)})}(\bbb{d})$ with $q \leq x$.	
\end{itemize}
\end{definition}
For a label $x \in \L^{(u)}(t)$ of some type $t$ at node $u \in \T$, we call $\bi$ a $x$-buddy for $t$ at $u$ if $\bi \in D_{-i}^{(u)}$ and $f_i(t, \bi)=x$.  We can now define extreme pivots.

\begin{definition}[Extreme Pivots]
Given a node $u$ and a pair of types $\bbb{1}, \bbb{2} \in D_i^{(u)}$, 
we say that types $\bbb{u}$, $\bbb{d}$ are \emph{extreme} pivots if they are pivots for $\bbb{1}$ and $\bbb{2}$ and, given $z \geq y > x \geq q$ as in Definition \ref{def:pivot}, we have that
	\begin{equation}\label{eq:twopivotcycle}
		w(P_1^{(out)}) + w(	P_2^{(in)})+w(	P_2^{(out)})+w(P_1^{(in)}) \geq 0, 
	\end{equation}
	for all paths $P_1^{(out)}$, $P_2^{(in)}$, $P_2^{(out)}$, $P_1^{(in)}$ in $\ver$ defined as follows:
	\begin{align*}
		P_1^{(out)} & := (\bbb{1}, \vect{b}^{(y)}_{-i}) \rightarrow \ab{1}   \rightsquigarrow \ab{k} \rightarrow  (\bbb{d}, \vect{b}^{(q)}_{-i})  \\
		P_2^{(in)} & := (\bbb{d}, \vect{b}^{(q)}_{-i}) \rightarrow \cc{1}    \rightsquigarrow  \cc{\ell}   \rightarrow  (\bbb{2}, \vect{b}^{(x)}_{-i})  \\
		P_2^{(out)} & := (\bbb{2}, \vect{b}^{(x)}_{-i}) \rightarrow \cc{\ell+1}  \rightsquigarrow \cc{\ell+h}   \rightarrow (\bbb{u}, \vect{b}^{(z)}_{-i}) \\
		P_1^{(in)} & := (\bbb{u}, \vect{b}^{(z)}_{-i}) \rightarrow \ab{k+1}  \rightsquigarrow \ab{k+g} \rightarrow (\bbb{1}, \vect{b}^{(y)}_{-i}),
	\end{align*}
	where $\vect{b}^{(y)}_{-i}$ is a $y$-buddy for $\bbb{1}$ at $u$, $\vect{b}^{(x)}_{-i}$ is an $x$-buddy for $\bbb{2}$ at $u$, $\vect{b}^{(q)}_{-i}$ is a $z$-buddy for $\bbb{d}$ at $v^{(d)}$,  $\vect{b}^{(z)}_{-i}$ is a $q$-buddy for $\bbb{u}$ at $v^{(u)}$, $\ab{1}, \ldots, \ab{k+g}$ and $\cc{1}, \ldots, \cc{\ell +h} $ are profiles in $\ver$.
\end{definition}

We will highlight below that property \eqref{eq:twopivotcycle} essentially implies that extreme pivots either have very small types or very high types with respect to $\bbb{1}$ and $\bbb{2}$ and outcomes $z \geq y > x \geq q$.

We are now ready to provide the definition of the mechanism format that 
characterizes OSP. 
\begin{definition}[Three-way greedy mechanism]
\label{def:3way}
A mechanism $\M=(f, p, \T)$ is three-way greedy if all its queries are ordered and for all internal nodes $u \in \T$ such that $i(u) \neq i$ and $\bbb{1}$ and $\bbb{2}$ in $D_i^{(u)}$ are antimonotone, it holds that any pair of pivots $\bbb{u}$ and $\bbb{d}$ are extreme.
\end{definition}

To make sense of the notion (and the name) of three-way greedy mechanisms, we now explore few of their properties. 
These mechanisms have essentially three characteristics.
Firstly, they only make ordered queries -- this is clearly linked to Theorem \ref{thm:monowlog}.
Secondly, they highlight a 4-point structure for the cycles of the \vgraph\ we should care about (while such a property was suggested by Theorem \ref{thm:anatomy}, Definition~\ref{def:3way} clarifies the profiles we need to be concerned with).

Before we move to the third feature of three-way greedy mechanisms, we provide some simple observations that build upon the 4-point structure of cycles. 
Firstly, we observe that three-way greedy mechanisms extend the two-way greedy mechanisms defined by \citet{fpv21}.
To this aim, given $\bbb{1}$ and $\bbb{2}$, with $\bbb{1}>\bbb{2}$, we name two pivots $\bbb{u}$ and $\bbb{d}$ for $\bbb{1}$ and $\bbb{2}$, \emph{anchors} if $z = y$ and $x = q$, with $x$, $y$, $z$, and $q$ being as in Definition \ref{def:pivot}.
Definition~\ref{def:3way} implies the following property.\footnote{Omitted and sketched proofs can be found in the appendix.}
\begin{observation}[No two anchors property]
\label{obs:no_anchors}
	If a mechanism $\M=(f, p, \T)$ is three-way greedy, then for each node $u \in \T$ with $i(u) \neq i$ such that a pair of antimonotone types $\bbb{1}$ and $\bbb{2}$ belong to $D_i^{(u)}$, there are no anchors for $\bbb{1}$ and $\bbb{2}$.
\end{observation}
\begin{proof}
Suppose that $\bbb{u}$ and $\bbb{d}$ are anchors for $\bbb{1}$ and $\bbb{2}$. 
In \eqref{eq:twopivotcycle}, consider $P_1^{(out)}=(\bbb{1}, \vect{b}^{(y)}_{-i}) \rightarrow (\bbb{d}, \vect{b}^{(q)}_{-i})$, $P_2^{(in)}=(\bbb{d}, \vect{b}^{(q)}_{-i}) \rightarrow (\bbb{2}, \vect{b}^{(x)}_{-i})$, $P_2^{(out)}=(\bbb{2}, \vect{b}^{(x)}_{-i}) \rightarrow (\bbb{u}, \vect{b}^{(z)}_{-i})$, $P_1^{(in)}=(\bbb{u}, \vect{b}^{(z)}_{-i}) \rightarrow (\bbb{1}, \vect{b}^{(y)}_{-i})$, where $q=x$, $z=y$, and $\vect{b}^{(y)}_{-i}$, $\vect{b}^{(z)}_{-i}$, $\vect{b}^{(x)}_{-i}$, and $\vect{b}^{(q)}_{-i}$ are as defined above. By definition of pivot, these edges all belong to $\ver$. 
By inspection, the cycle composed by these four edges has weight
$
(y - x) (\bbb{2} - \bbb{1}) <0,
$
thus contradicting \eqref{eq:twopivotcycle}.
\end{proof}

With only two outcomes available, the only pivots possible must have outcomes that are equal to those of the two antimonotone types, i.e., they are anchors. This leads to negative-weight cycles, as shown by Observation~\ref{obs:no_anchors}. In fact, two-way greedy mechanisms interact with each agent either in a greedy fashion (by querying about the best type that has not yet been  queried, and in case of positive answer, by including her in the solution) or in a reverse greedy fashion (by asking her whether her type is the worst that has not yet been queried, and in case of positive answer, by excluding her from the solution) as long as there are still antimonotone types. This way, the mechanism completely avoids the chance of having two anchors. 

The intuition arising from the binary allocation case, can then be extended to a more general outcome space by considering not only anchors, but generic pivots. Indeed, with more than two outcomes there may be pivots whose label is different from those of antimonotone types: however, as for the binary case,
by avoiding the introduction of 
pivots until there are antimonotone types will surely satisfy Definition~\ref{def:3way}.
How can a mechanism avoid two pivots in this setting? Clearly, the mechanism can query an agent in a greedy or reverse greedy fashion exactly as in the binary allocation case (except that now we are allowed to allocate different outcomes at each query, whilst maintaining the monotonicity with respect to types in order to not violate 2CMON). The third possibility is for the mechanism to first query an agent about whether her type is large or small (with the exact threshold defining large or small types depending on the problem at the hand), whilst ensuring that a label for a large type is never better than the label of a small type, and then in case of large types, proceeding by querying the agent greedly, whereas in the case of small types, the mechanism queries the agent in a reverse greedy fashion. These three ways of ensuring that no two pivots exist justify the name of our mechanism (see also Example \ref{example} in Appendix).

\begin{example}\label{example}
In order to clarify how these three-way greedy mechanisms can work, let us consider a very simple example for the aforementioned scheduling related machine problem 
with two jobs and two 
machines, whose types belong to 
$\{T, H, L, B\}$, with $T > 2H$, $H > 2L$, and $L > 2B$.

A three-way greedy mechanism can greedily start by asking machine 1 whether her type is $B$ and, for example, allocating both jobs in case of a positive answer. Note that in this case successive queries to machine 1 (that may be preceded by queries to machine 2) we must ask for the next smaller unqueried type (thus, $L$ first and $H$ afterwards) until a positive answer is received by allocating for each positive answer to these queries a number of jobs that is at most as large as the number of jobs allocated at the previous query (thus, for example, when querying $L$ the mechanism may offer both jobs or only $1$, but in the latter case, when querying $H$ the mechanism cannot offer both jobs to machine 1).

Alternatively, the three-way greedy mechanism can start in a reverse greedy fashion by asking to machine 1 whether her type is $T$ and, e.g., allocating no job in case of a positive answer. Similarly to the previous case, successive queries to machine 1 
must ask for the next larger unqueried type 
until a positive answer is received. For each positive answer, the mechanism must allocate a number of jobs that is no smaller than the number of jobs allocated at the previous query (thus, for example, for a query $H$ the mechanism may offer either no jobs or only $1$ job, but in the latter case, for the query $L$ machine 1 needs to get at least one job).

Finally, the three-way mechanism can start by asking whether the type of machine 1 is high, i.e., it is in $\{T, H\}$, or low, i.e., it is in $\{L, B\}$, by assuring that for high types the number of assigned jobs is no larger then the job load for low types (e.g., at most one job is assigned for high types, and at least one for low types). If machine 1 declared to have a high type, then the next query 
must ask if the type is $H$ and assign (in case of positive answer) a number of jobs that is not smaller than the load assigned in case the machine would have a larger type. Similarly, if machine 1 chose an action signalling a low type, then the next query 
must ask if the type is $L$ and assign (in case of positive answer) a number of jobs that is not larger than the load assigned in case the machine would have a smaller type.
\end{example}

The third feature that emerges from the definition of three-way greedy mechanisms is that 
pivots can indeed exist, as long as the one with small (large) label is large (small) \emph{enough}. Here, the thresholds for these pivots to be considered small/large enough depend on cycles that go through the four aforementioned points, cf. the definitions of the paths in the \vgraph. 
\begin{observation}[Size of extreme pivots]
\label{obs:extreme_pivots}
    If a mechanism $\M=(f, p, \T)$ is three-way greedy, then for each node $u \in \T$ with $i(u) \neq i$ such that a pair of antimonotone types $\bbb{1}$ and $\bbb{2}$, with $\bbb{1} > \bbb{2}$ belong to $D_i^{(u)}$, if there are pivots $\bbb{u}$ and $\bbb{d}$ for $\bbb{1}$ and $\bbb{2}$, then it must be the case that $\bbb{d} > \bbb{1}$, $\bbb{2} > \bbb{u}$ and
       $ \Lambda (\bbb{d} - \bbb{1}) + \Delta (\bbb{2} - \bbb{u}) \geq \delta (\bbb{1} - \bbb{2})$,
    where $\Delta = (z - y)$, $\Lambda = (x-q)$, $\delta = (y - x)$, and $z \geq y > x \geq q$ are as in  Definition \ref{def:pivot}.
\end{observation}
\begin{proof}
    The ranking of types simply follows since, by definition of pivot, $\bbb{d}$ is separated from $\bbb{1}$, $\bbb{2}$ is separated from $\bbb{u}$, there is a $q$-buddy $\vect{b}^{(q)}_{-i}$ for $\bbb{d}$, a $y$-buddy $\vect{b}^{(y)}_{-i}$ for $\bbb{1}$, a $z$-buddy $\vect{b}^{(z)}_{-i}$ for $\bbb{u}$, a $x$-buddy $\vect{b}^{(x)}_{-i}$ for $\bbb{2}$ at the respective divergence nodes, and the mechanism $\M$ satisfies OSP 2CMON, hence it never separates antimonotone types.
    
    In \eqref{eq:twopivotcycle}, let us consider $P_1^{(out)}=(\bbb{1}, \vect{b}^{(y)}_{-i}) \rightarrow (\bbb{d}, \vect{b}^{(q)}_{-i})$, $P_2^{(in)}=(\bbb{d}, \vect{b}^{(q)}_{-i}) \rightarrow (\bbb{2}, \vect{b}^{(x)}_{-i})$, $P_2^{(out)}=(\bbb{2}, \vect{b}^{(x)}_{-i}) \rightarrow (\bbb{u}, \vect{b}^{(z)}_{-i})$, $P_1^{(in)}=(\bbb{u}, \vect{b}^{(z)}_{-i}) \rightarrow (\bbb{1}, \vect{b}^{(y)}_{-i})$. By definition of pivot, these edges all belong to $\ver$.
Hence, we have that  
$
\bbb{1}(q - y) + \bbb{d} (x - q) + \bbb{2} (z - x) + \bbb{u} (y - z) \geq 0,
$
from which the claim 
immediately follows.
\end{proof}
In order to appreciate how demanding the condition above 
is, let us focus on some simple cases. Suppose first that $\Lambda = 0$, i.e., $\bbb{d}$ is an anchor. Then we have that
$
 \bbb{u} \leq \bbb{2} - \frac{\delta}{\Delta} (\bbb{1} - \bbb{2}).
$ 
Similarly, if $\Delta = 0$, i.e., $\bbb{u}$ is an anchor, then we have that
$
 \bbb{d} \geq \bbb{1} + \frac{\delta}{\Lambda} (\bbb{1} - \bbb{2}). 
$ 
This already highlights that either the difference between the outcomes of pivots and the outcome of antimonotone types is large (i.e., $\Delta$ or $\Lambda$ are large with respect to $\delta$) or the difference between pivot types and antimonotone types needs to be large enough (at least $\frac{\delta}{X} (\bbb{1} - \bbb{2})$, $X$ being either $\Delta$ or $\Lambda$). To some extent, these pivots are \emph{clearly separable} from antimonotone types in the sense that either their outcomes or their values are much different from the ones of antimonotone types. 
We also observe that the gap $\frac{\delta}{X} (\bbb{1} - \bbb{2})$ that emerges above can be 
a very loose bound. Indeed, while in Observation~\ref{obs:extreme_pivots} we considered $P_1^{(out)}$, $P_2^{(in)}$, $P_2^{(out)}$, and $P_1^{(in)}$ consisting of a single edge, there are nodes in the implementation trees belonging to paths among the corresponding pairs of profiles whose weight can in principle be much smaller than the weight of this single edge. Hence, by replacing our choice of $P_1^{(out)}$, $P_2^{(in)}$, $P_2^{(out)}$, and $P_1^{(in)}$ with these lighter paths we are increasing the required gap between the pivots and the antimonotone types.

\paragraph{The Proof.}
We begin by understanding the anatomy of antimonotone types and characterize the shape of the implementation tree that leads to such a pair. Let $u$ denote the deepest node in the implementation tree of a 2CMON mechanism $\M =(f, p, \T)$ such that two antimonotone types $\bbb{1}$ and $\bbb{2}$ exist whose witnesses $\bb{1}, \bb{2}$ belong to $D^{(u)}$. This means that $\bb{1}$ and $\bb{2}$ will be separated at $u$. We call $u$ the \emph{divergence node} of $\bb{1}$ and $\bb{2}$.

By 2CMON, we can observe that $i(u)=j \neq i$ and $\bbbb{1}{j} \in L$, $\bbbb{2}{j} \in R$, where $L$ and $R$ are the parts defined by the query at $u$. We call all the profiles that can be reached via $L$ ($R$, respectively) the $\bbb{1}$ ($\bbb{2}$, respectively) peers. 
By definition of \vgraph, it is not hard to see the following property.

\begin{observation}[No crossing property]\label{obs:nocross}
	Let $\bbb{1},\bbb{2}$ be two antimonotone types witnessed by $\bb{1}$ and $\bb{2}$.	In $\ver$, there is no edge between $\bbb{1}$ peers and $\bbb{2}$ peers.
\end{observation}

\begin{theorem}\label{thm:main}
	The mechanism $\M=(f, p, \T)$ is  OSP for single-parameter domains if and only if  $\M$ is three-way greedy.
\end{theorem}
\begin{proof}
	Let us start from the only if part. We here simply need to observe that in \eqref{eq:twopivotcycle} we are bounding from below the weight of a cycle of $\ver$; therefore, the cycle must be non-negative as requested since the mechanism is OSP.
	
	As for the sufficiency, we firstly observe that any three-way greedy mechanism is OSP 2CMON by construction. Assume by contradiction that the mechanism is not OSP, and let $C$ be the negative cycle in $\ver$. By Theorem \ref{thm:anatomy} there is at least one pair of antimonotone types $\bbb{1}$ and $\bbb{2}$ with two witnessing profiles $\bb{1}$ and $\bb{2}$ in $C$. 
	By the no crossing property (Observation \ref{obs:nocross}) the only way to go from $\bb{1}$ to $\bb{2}$ and back is through (at least) two different profiles (w.l.o.g., $C$ is a simple cycle) that have been separated by $i$ above the divergence node of $\bb{1}$ and $\bb{2}$. Since by definition the mechanism uses ordered queries, the outcomes for these nodes are either at most $\max\{f_i(\bb{1}), f_i(\bb{2})\}$ or at least $\min\{f_i(\bb{1}), f_i(\bb{2})\}$. But then these two profiles are pivots and since $w(C)<0$, this cycle would contradict \eqref{eq:twopivotcycle}. 
\end{proof}

\section{Proving that Ordered Queries are Necessary}\label{sec:ordered}
In this section, we will prove Theorem \ref{thm:monowlog}. We first give a host of novel tools to bound path weights. 
In all the following lemmas, we consider a mechanism $\M$ with implementation tree $\T$ that satisfies 2CMOM but not CMON, that is, for some player $i$, all two-cycles are positive but a negative longer cycle exists in $\ver$. We consider different ways in which this cycle can be modified.

For a sequence of profiles $(\x^1, \ldots, \x^\ell)$ and an agent $i$ we 
define 
$w_i(\x^1, \ldots, \x^\ell) = \sum_{j=1}^{\ell-1} x_i^j (f_i(\x^{j+1}) - f_i(\x^j)).$ 
Observe that if the sequence of profiles $(\x_1, \ldots, \x_\ell)$ is a path in $\ver$ (that is, $(\x_i, \x_{i+1})$ is an edge of $\ver$), then $w_i(\x_1, \ldots, \x_\ell)$ is exactly the cost of the path.


\begin{figure}
     \centering
     \begin{subfigure}[b]{0.32\textwidth}
         \centering
         \includegraphics[scale=0.6]{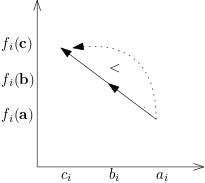}
         \caption{Expansion: Two-edge paths going leftwards and upwards in the space weigh less than the direct edge between the endpoints.}
         \label{fig:exp}
     \end{subfigure}
     \hfill
     \begin{subfigure}[b]{0.32\textwidth}
         \centering
         \includegraphics[scale=0.6]{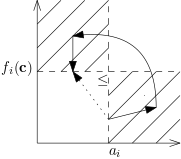}
         \caption{Reduction: Paths using profiles in dashed areas of the space weigh more than the direct edge between the endpoints.}
         \label{fig:red}
     \end{subfigure}
     \hfill
    \begin{subfigure}[b]{0.32\textwidth}
         \centering
         \includegraphics[scale=0.6]{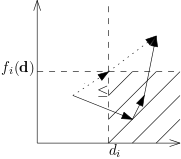}
         \caption{Replacement: Paths using profiles in dashed areas of the space weigh more than the 2-edge path between the endpoints.}
         \label{fig:repl}
     \end{subfigure}    
        \caption{Examples of our ironing techniques in a bidimensional space where the $x$-axis contains possible types and the $y$-axis contains outcomes computed by the social choice function.}
        \label{fig:ironing}
\end{figure}

We give three different operations to ``iron'' a path in a \vgraph\ that satisfies 2CMON but not CMON. Figure \ref{fig:ironing} gives a visual representation of these operations. Our forms of ironing abstract the notion of path and actually work on sequences of profiles and our definition of $w_i$, 
meaning that the (ironed) path may not in general belong to the \vgraph\ (for all but one operation, we make no assumption on the existence of the edges of the ironed path). Clearly, these properties and bounds will be applied to the \vgraph\ only when edges exist and their application is useful for our purposes.

The first ironing operation shows under which conditions we can add edges to a path without changing the sign of the weight of a negative cycle. 

\begin{lemma}[Positive path expansion]
\label{lem:bypass}
  Suppose there is a negative cycle $C$ in a \vgraph\ $\ver$, but no negative 2-cycles. If $C$ uses an edge $(\a,\c)$ such that $f_i(\a) \leq f_i(\c)$, and there is a profile $\b$ such that edges $(\a,\b)$ and $(\b,\c)$ exist in $\ver$ with $a_i \geq b_i \geq c_i$, then the cycle obtained by replacing edge $(\a,\c)$ with the path $\a \rightarrow \b \rightarrow \c$ is still negative.
\end{lemma}
\begin{proof}
 Let $C$ be the original cycle, and let $C'$ be the cycle obtained by replacing the edge $(\a,\c)$ in $C$ with the path $\a \rightarrow \b \rightarrow \c$. We next evaluate the difference between the cost of $C$ and the cost of $C'$. To this aim
let $\partial C$ be the set of all edges in $C$ different from $(\a,\c)$.
Since $w(C) = a_i(f_i(\c) - f_i(\a)) + \sum_{(\x,\y) \in \partial C} w(\x,\y)$ and $w(C') = a_i(f_i(\b) - f_i(\a)) + b_i(f_i(\c) - f_i(\b)) + \sum_{(\x,\y) \in \partial C} w(\x,\y)$, then
$$
 w(C) - w(C') = (f_i(\c)-f_i(\b)) (a_i - b_i) > 0,
$$
where the inequality follows since, by hypothesis, $a_i > b_i$ and $f_i(\c) \geq f_i(\b)$ (otherwise, 2CMON would be violated since the edge $(\b,\c)$ exists and $c_i < b_i$). We can then conclude that $w(C') < w(C) < 0$, as desired.
\end{proof}

We next show when we can substitute a path with a direct edge between its endpoints. This form of ironing comes in three different versions, of increasing strength. 

\begin{lemma}[Weak path reduction]
\label{lem:angle_redux_step1}
 Let $(\a, \b^1, \ldots, \b^\ell, \c)$ be a sequence of profiles.
 If $a_i \leq b^1_i \leq \cdots \leq b^\ell_i$ and $f_i(\c) \geq \max_{j \in \ell} f_i(\b^j)$,
 then $w_i(\a, \b^1, \ldots, \b^\ell, \c) \geq w_i(\a,\c)$.
 
 Similarly, if $a_i \geq b^1_i \geq \cdots \geq b^\ell_i$ and $f_i(\c) \leq \min_{j \in \ell} f_i(\b^j)$, then $w_i(\a, \b^1, \ldots, \b^\ell, \c) \geq w_i(\a,\c)$.
 \end{lemma}
\begin{proof}
Observe that
$$w_i(\a, \b^1, \ldots, \b^\ell, \c) = a_i(f_i(\b^1) - f_i(\a)) + \sum_{j=1}^{\ell-1} b^j_i(f_i(\b^{j+1}) - f_i(\b^j)) + b^\ell_i(f_i(\c) - f_i(\b^\ell))$$
and $w_i(\a,\c) = a_i(f_i(\c) - f_i(\a))$. Then
$$
 w_i(\a, \b^1, \ldots, \b^\ell, \c) - w_i(\a, \c) = f_i(\c) (b^\ell_i - a_i) - \sum_{j=2}^\ell f_i(\b^j) (b_i^j - b_i^{j-1}) - f_i(\b^1) (b_i^1 - a_i).
$$

Suppose first that $a_i \geq b^1_i \geq \cdots \geq b^\ell_i$ and $f_i(\c) \leq f^- = \min_{j \in \ell} f_i(\b^j)$.
Then, it follows that 
\begin{align*}
 w_i(\a, \b^1, \ldots, \b^\ell, \c) - w_i(\a, \c) & = f_i(\b^1) (a_i - b_i^1) + \sum_{j=2}^\ell f_i(\b^j) (b_i^{j-1} - b_i^j) - f_i(\c) (a_i - b^\ell_i)\\
 & \geq \left(f^- - f_i(\c)\right) (a_i - b^\ell_i) \geq 0,
\end{align*}
where the first inequality follows since $f_i(\x) \geq 0$ for every outcome $\x$, and all the differences among types are non-negative by hypothesis, and the second inequality follows since, by hypothesis, $f^- > f_i(\c)$ and $a_i \geq b^\ell_i$.

Consider now the case that $a_i \leq b^1_i \leq \cdots \leq b^\ell_i$ and $f_i(\c) \geq f^+ = \max_{j \in \ell} f_i(\b^j)$.
Then, it follows that 
\begin{align*}
 w_i(\a, \b^1, \ldots, \b^\ell, \c) - w_i(\a, \c) & \geq f_i(\c) (b^\ell_i - a_i) - \sum_{j=2}^\ell f^+ (b_i^j - b_i^{j-1}) - f^+(\b^1) (b_i^1 - a_i)\\
 & = \left(f_i(\c) - f^+\right) (b^\ell_i - a_i) \geq 0,
\end{align*}
where the first inequality follows since $f_i(\x) \geq 0$ for every outcome $\x$, and all the differences among types are non-negative by hypothesis, and the second inequality follows since $f_i(\c) \geq f^+$ and $b^\ell_i \geq a_i$,   by hypothesis.
\end{proof}

\begin{lemma}[Path reduction]
\label{lem:angle_redux_step2}
 Let $(\a, \b^1, \ldots, \b^{\ell}, \c)$ be a sequence of profiles.
 If $a_i \leq b_i^j$ and $f_i(\c) \geq f_i(\b^j)$ for every $j \in [\ell]$,
 then $w_i(\a, \b^1, \ldots, \b^{\ell}, \c) \geq w_i(\a,\c)$.
 
 Similarly, if $a_i \geq b_i^j$ and $f_i(\c) \leq f_i(\b^j)$ for every $j \in [\ell]$,
 then $w_i(\a, \b^1, \ldots, \b^{\ell}, \c) \geq w_i(\a,\c)$.
 \end{lemma}
\begin{proof}
Let us first consider the case that $a_i \leq b_i^j$ and $f_i(\c) \geq f_i(\b^j)$ for every $j \in [\ell]$.

The proof is by induction on $\ell$. If $\ell = 1$, then the sequence $(\a, \b^1, \c)$ satisfies the conditions of Lemma~\ref{lem:angle_redux_step1} and the claim follows.
Suppose now, by inductive hypothesis, that the claim holds for every sequence of profiles with $\ell + 2$ nodes that satisfy the conditions that $a_i \leq b_i^j$ and $f_i(\c) \geq f_i(\b^j)$ for every $j \in [\ell]$. Consider a sequence of profiles $(\a, \b^1, \ldots, \b^{\ell+1}, \c)$. If $b_i^1 \leq \cdots \leq b_i^{\ell+1}$, then the conditions of Lemma~\ref{lem:angle_redux_step1} are satisfied and the claim follows. Otherwise there is $j \in \{2, \ldots, \ell-1\}$ such that $b_i^j$ is in between $b_i^{j-1}$ and $b_i^{j+1}$. However, the sequence of profiles $(\b^{j-1}, \b^j, \b^{j+1})$ satisfies the conditions of Lemma~\ref{lem:angle_redux_step1}, and by applying the lemma, we obtain that $w_i(\b^{j-1}, \b^j, \b^{j+1}) \geq w_i(\b^{j-1}, \b^{j+1})$.
Hence, it immediately follows that $w_i(\a, \b^1, \ldots, \b^{\ell+1}, \c) \geq w_i(\a, \b^1, \ldots, \b^{j-1}, \b^{j+1}, \ldots, \b^{\ell+1}, \c) \geq w_i(\a,\c)$, where the last inequality follows from the inductive hypothesis, since $(\a, \b^1, \ldots, \b^{j-1}, \b^{j+1}, \ldots, \b^{\ell+1}, \c)$ has length $\ell + 2$ and satisfies the necessary conditions.

The proof for the remaining case, i.e., when $a_i \geq b_i^j$ and $f_i(\c) \leq f_i(\b^j)$ for every $j \in [\ell]$, is very similar and hence omitted. 
%
\end{proof}

We can strengthen the path reduction lemma even further. Towards this end, let $(\a, \b^1, \ldots, \b^{\ell}, \c)$ be a sequence of profiles. We say that $(\a, \b^1, \ldots, \b^{\ell}, \c)$ is \emph{good} if for every $j \in [\ell]$ either 
 \emph{(i)} $a_i \leq b_i^j$ and $f_i(\c) \geq f_i(\b^j)$; or, \emph{(ii)} 
 $a_i \geq b_i^j$ and $f_i(\c) \leq f_i(\b^j)$. \label{cond:2}

\begin{lemma}[Strong path reduction]
\label{lem:angle_redux_step3}
 If $(\a, \b^1, \ldots, \b^{\ell}, \c)$ is a good sequence of profiles, then $$w_i(\a, \b^1, \ldots, \b^{\ell}, \c) \geq w_i(\a,\c).$$
 \end{lemma}
\begin{proof}
Consider a good sequence of profiles $(\a, \b^1, \ldots, \b^{\ell}, \c)$. For $j \in [\ell]$, we say that $j$ is a \emph{swing} index if condition~\emph{(i)} 
 holds for $j$ but not for $j+1$, or vice versa. 
The proof is by induction on the number $k$ of swing indices in $(\a, \b^1, \ldots, \b^{\ell}, \c)$. If $k = 0$, then this sequence of profiles satisfies the conditions of Lemma~\ref{lem:angle_redux_step2} since it is good and the claim follows.

Suppose now, by inductive hypothesis, that the claim holds for every good sequence of profiles with $k$ swing indices, and consider a good sequence of profiles $(\a, \b^1, \ldots, \b^{\ell}, \c)$ with $k+1$ swing indices. Let $j$ be the first swing vertex in $(\a, \b^1, \ldots, \b^{\ell}, \c)$. Then we have that $(\a, \b^1, \ldots, \b^{j}, \b^{j+1})$ satisfies the conditions of Lemma~\ref{lem:angle_redux_step2}, and thus $w_i(\a, \b^1, \ldots, \b^{j}, \b^{j+1}) \geq w_i(\a,\b^{j+1})$. Hence it follows that $w_i(\a, \b^1, \ldots, \b^{\ell}, \c) \geq w_i(\a, \b^{j+1}, \ldots, \b^{\ell}, \c) \geq w_i(\a,\c)$, where the last inequality follows from the inductive hypothesis, since $(\a, \b^{j+1}, \ldots, \b^{\ell}, \c)$ is a good sequence of profiles with $k$ swing indices.
\end{proof}

\begin{corollary}[Path reduction in the \vgraph]\label{cor:angle_redux}
 If in $\ver$ there is a negative cycle that uses a path $P:=\a \rightarrow \b^1  \rightsquigarrow \b^\ell \rightarrow \c$ that is a good sequence of profiles, and the edge $(\a,\c)$ exists in $\ver$, then the cycle obtained by replacing the path $P$ with the edge $(\a, \c)$ is still negative.
 \end{corollary}
 
The third ironing operation allows to replace a path with two edges; we call this path replacement. 
 
\begin{lemma}[Weak path replacement]
\label{lem:angle_replace_step1}
Let $(\a, \b^1, \ldots, \b^\ell, \c)$ be a sequence of profiles and $\d$ be a profile different from $\a$ and $\c$.
If $a_i \leq d_i \leq b^1_i \leq \cdots \leq b^\ell_1$, $f_i(\c) \geq f_i(\d) \geq \max_{j \in \ell} f_i(\b^j)$,
then $w_i(\a, \b^1, \ldots, \b^\ell, \c) \geq w_i(\a,\d,\c)$. 
Similarly, if $a_i \geq d_i \geq b^1_i \geq \cdots \geq b^\ell_1$, $f_i(\c) \leq f_i(\d) \leq \min_{j \in \ell} f_i(\b^j)$,
then $w_i(\a, \b^1, \ldots, \b^\ell, \c) \geq w_i(\a,\d,\c)$.
\end{lemma}
\begin{proof}
Observe that
$$w_i(\a, \b^1, \ldots, \b^\ell, \c) = a_i(f_i(\b^1) - f_i(\a)) + \sum_{j=1}^{\ell-1} b^j_i(f_i(\b^{j+1}) - f_i(\b^j)) + b^\ell_i(f_i(\c) - f_i(\b^\ell))$$
and 
$w_i(\a,\d,\c) = a_i(f_i(\d) - f_i(\a)) + d_i(f_i(\c) - f_i(\d)).$ 
Let $\Delta w_i = w_i(\a, \b^1, \ldots, \b^\ell, \c) - w_i(\a, \d, \c)$. Then we have that
$$
 \Delta w_i = f_i(\c)(b_i^\ell - d_i) - f_i(\b^1)(b_i^1 - a_i)- \sum_{j=2}^\ell f_i(\b^j)(b_i^j - b_i^{j-1}) + f_i(\d)(d_i - a_i).
$$

Consider first the case that $a_i \geq d_i \geq b^1_i \geq \cdots \geq b^\ell_i$ and $f_i(\c) \leq f_i(\d) \leq f^-$.
Then, we get:
\begin{align*}
 \Delta w_i & = f_i(\b^1)(a_i - b_i^1) + \sum_{j=2}^\ell f_i(\b^j)(b_i^{j-1} - b_i^{j}) - f_i(\d)(a_i - d_i) - f_i(\c)(d_i - b_i^\ell)\\
 & \geq f^- (a_i - b_i^\ell) - f_i(\d)(a_i - d_i) - f_i(\c)(d_i - b_i^\ell)\\
 & = (f^- - f_i(\d)) (a_i - b_i^\ell) + (f_i(\d)-f_i(\c))(d_i - b_i^\ell) \geq 0,
\end{align*}
where the first inequality follows since $f_i(\x) \geq 0$ for every outcome $\x$, and all the differences among types are non-negative by hypothesis, and the second inequality follows since $f^- \geq f_i(\d) \geq f_i(\c)$ and $a_i \geq d_i \geq b_i^\ell$ by hypothesis.

Consider now the case that $a_i \leq d_i \leq b^1_i \leq \cdots \leq b^\ell_i$ and $f_i(\c) \geq f_i(\d) \geq f^+ = \max_{j \in \ell} f_i(\b^j)$:
\begin{align*}
 \Delta w_i & \geq f_i(\c) (b^\ell_i - d_i) - f^+(\b^1) (b_i^1 - a_i) - \sum_{j=2}^\ell f^+ (b_i^j - b_i^{j-1}) + f_i(\d)(d_i - a_i)\\
 & = f_i(\c) (b^\ell_i - d_i) - f^+ (b_i^\ell - a_i) + f_i(\d)(d_i - a_i)\\
 & = (f_i(\c)-f_i(\d))(b_i^\ell - d_i) + (f_i(\d) - f^+) (b_i^\ell - a_i) \geq 0,
\end{align*}
where the first inequality follows since $f_i(\x) \geq 0$ for every outcome $\x$, and all the differences among types are non-negative by hypothesis, and the second inequality follows since $f_i(\c) \geq f_i(\d) \geq f^+$ and $a_i \leq d_i \leq b_i^\ell$ by hypothesis.
\end{proof}

It is easy to see that next lemma immediately follows by applying the proofs of Lemma~\ref{lem:angle_redux_step2} and Lemma~\ref{lem:angle_redux_step3}, respectively, and simply invoking Lemma~\ref{lem:angle_replace_step1} in place of Lemma~\ref{lem:angle_redux_step1}.

\begin{lemma}[Path replacement]
\label{lem:angle_replace_step2}
 Let $(\a, \b^1, \ldots, \b^{\ell}, \c)$ be a sequence of profiles and $\d$ be a profile different from $\a$ and $\c$.
 If $a_i \leq d_i \leq b_i^j$ and $f_i(\c) \geq f_i(\d) \geq f_i(\b^j)$ for every $j \in [\ell]$,
 then $w_i(\a, \b^1, \ldots, \b^{\ell}, \c) \geq w_i(\a, \d, \c)$. 
 Similarly, if $a_i \geq d_i \geq b_i^j$ and $f_i(\c) \leq f_i(\d) \leq f_i(\b^j)$ for every $j \in [\ell]$,
 then $w_i(\a, \b^1, \ldots, \b^{\ell}, \c) \geq w_i(\a, \d, \c)$.
 \end{lemma}
\begin{corollary}[Path replacement in the \vgraph]
\label{cor:angle_replace}
 Suppose that in $\ver$ there is a negative cycle that uses a path $P:=\a \rightarrow \b^1  \rightsquigarrow \b^\ell \rightarrow \c$ and there is a profile $\d$ different from $\a$ and $\c$, such that (i) $a_i \leq d_i \leq b_i^j$ and $f_i(\c) \geq f_i(\d) \geq f_i(\b^j)$ for every $j \in [\ell]$; (ii) edges $(\a,\d)$ and $(\d,\c)$ exists in $\ver$; then the cycle obtained by replacing $P$ with $\a \rightarrow \d \rightarrow \c$ is still negative. 
 Similarly, if in $\ver$ there is a negative cycle that uses a path $P:=\a \rightarrow \b^1  \rightsquigarrow \b^\ell \rightarrow \c$, and there is a profile $\d$ different from $\a$ and $\c$, such that (i) $a_i \geq d_i \geq b_i^j$ and $f_i(\c) \leq f_i(\d) \leq f_i(\b^j)$ for every $j \in [\ell]$; (ii) edges $(\a,\d)$ and $(\d,\c)$ exists in $\ver$; then the cycle obtained by replacing $P$ with $\a \rightarrow \d \rightarrow \c$ is still negative.
 \end{corollary}


We want to prove that there is an OSP mechanism $\M=(f,p, \T)$ 
if and only if there is an ordered OSP mechanism $\M'=(f,p, \T')$. That is, we can morph $\T$ into $\T'$ where each query is ordered. 

One direction is trivial. Hence, we next prove that if $\M$ is a non-ordered OSP mechanism, then there must be a mechanism $\M'$ which returns the same outcomes, is OSP and is ordered. 

Specifically, let $u$ be a node of the implementation tree $\T$ of $\M$ where the mechanism performs a non-ordered query, i.e., $\max \{t \in L\} > \min \{t \in R\}$. We will denote with $\T_L$ and $\T_R$ the subtrees corresponding to the action signalling types in $L$ and $R$, respectively. By convention, we name $L$ the side containing $\min \{t \in D_i^{(u)}\}$.

Let $\T'$ be the implementation tree $\T$ of $\M$ obtained by replacing the query at such a node $u$ with a sequence of queries defined as follows. Let us partition $D_i(u)$ in:
$L_1 = \left\{t \in L \mid t < \min \{t \in R\}\right\}$; moreover for every $j \geq 1$,
$$R_j = \left\{t \in R \setminus \bigcup_{k = 1}^{j-1} R_k \mid t < \min \left\{t \in L \setminus \bigcup_{k = 1}^j L_k\right\}\right\},$$
and for every $j \geq 2$
$$L_j = \left\{t \in L \setminus \bigcup_{k = 1}^{j-1} L_k \mid t < \min \left\{t \in R \setminus \bigcup_{k = 1}^{j-1} R_k\right\}\right\}.$$
The mechanism $\M'$ then first separates $L_1$ from the rest of domain, then $R_1$ from the rest, and so on, by attaching, for every $j$, $\T_R$ and $\T_L$ to the action corresponding to $R_j$ and $L_j$, respectively. An example of this construction is shown in Figure \ref{fig:localmod}. 

\begin{figure*}[t!]
		\centering
		$  \displaystyle
		D_i^{(u)} = \{\underbrace{l_1<\cdots<l_\alpha}_{L_1} < \underbrace{r_1<\cdots<r_\beta}_{R_1} < \underbrace{l_{\alpha+1}<\cdots<l_\gamma}_{L_2} < \cdots < \underbrace{r_\delta<\cdots<r_\rho}_{R_3} \}
		$
		
		\begin{subfigure}[t]{0.41\linewidth}
			\centering
			\begin{tikzpicture}[shorten >=1pt,node distance=1.5cm,inner sep=1pt]
				\tikzstyle{place}=[circle,draw]
				\tikzstyle{placew}=[rectangle]
				
				\node[place] (x)      {$u$};
				
				\node[placew] (a)   [below left of=x]  {$\mathcal{T}_L$};
				\node[placew] (b)   [below right of=x]  {$\mathcal{T}_R$};
				
				\path[-] (x) edge node[left] {$L=\{l_1, \ldots, l_\lambda\}\;$} (a);
				\path[-] (x) edge node [right] {$\;R=\{r_1, \ldots, r_\rho\}$} (b);
			\end{tikzpicture}
		\end{subfigure}
		\qquad $\Longrightarrow$
		\begin{subfigure}[t!]{0.49\linewidth}
			\centering
			\begin{tikzpicture}[shorten >=1pt,node distance=1.5cm,inner sep=1pt]
				\tikzstyle{place}=[circle,draw]
				\tikzstyle{placew}=[rectangle]
				\node[place] (x)  		     {$v_1$};
				\node[place] (y)   [below left of=x]  {$v_2$};
				\node[place] (z) [below left of=y]  {$v_3$};
				\node[place] (w) [below left of=z]  {$v_4$};
				\node[place] (s) [below left of=w]  {$v_5$};
				
				\node[placew] (a)   [below right of=x]  {$\mathcal{T}_L$};
				\node[placew] (b)   [below right of=y]  {$\mathcal{T}_R$};
				\node[placew] (d)   [below right of=z]  {$\mathcal{T}_L$};
				\node[placew] (c)   [below right of=w]  {$\mathcal{T}_R$};
				\node[placew] (e)   [below right of=s]  {$\mathcal{T}_L$};
				\node[placew] (f)   [below left of=s]  {$\mathcal{T}_R$};
				
				\path[-] (x) edge node[left] {$D_i^{(v_1)}\setminus L_1\;\;$} (y);
				\path[-] (x) edge node [right] {$\;L_1$} (a);
				\path[-] (y) edge node [left] {$D_i^{(v_2)}\setminus R_1\;\;$} (z);
				\path[-] (y) edge node [right] {$\;R_1$} (b);
				\path[-] (z) edge node [left] {$D_i^{(v_3)}\setminus L_2\;\;$} (w);
				\path[-] (z) edge node [right] {$\;L_2$} (d);
				\path[-] (w) edge node [right] {$\;R_2$} (c);
				\path[-] (w) edge node [left] {$D_i^{(v_4)}\setminus R_2\;\;$} (s);
				\path[-] (s) edge node [right] {$\;L_3$} (e);
				\path[-] (s) edge node [left] {$R_3\;\;$} (f);
			\end{tikzpicture}
		\end{subfigure}
	\caption{An example of the way we order queries in Theorem \ref{thm:monowlog}; the original tree $\T$ is on the left, whilst the modified tree $\T'$ is on the right. The domain of $i=i(u)$ at $u \in \T$ is at the top.}
	\label{fig:localmod}
	\end{figure*}
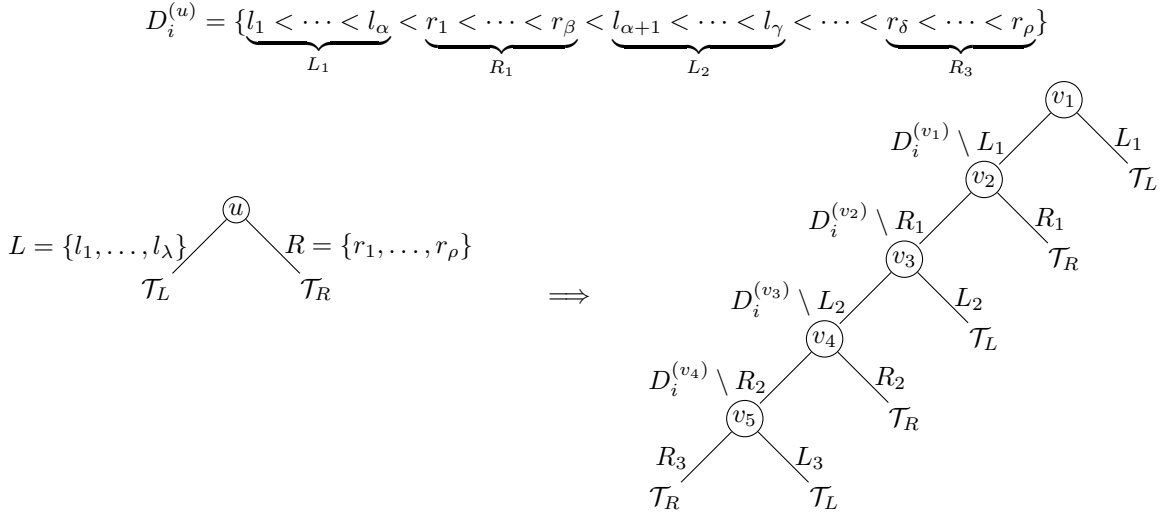

We next prove the following lemma, that completes the proof of Theorem~\ref{thm:monowlog}.
\begin{lemma}
The mechanism $\M'=(f, p, \T')$ is an ordered OSP.
\end{lemma}
\begin{proof}
 It is clear from the way we modified $\T$ that $\T'$ uses queries that separate smaller types from larger types. To prove that the mechanism $\M'$ is ordered it suffices to show that $\M'$ satisfies 2CMON. The only new edges that may in principle violate 2CMON are those between types $p < p'$ that in $\M$ were in the same part $P$ at divergence at $u\in \T$ that are now separated by $\T'$ because of our modification (e.g., $l_\alpha$ and $l_\gamma$ in Figure \ref{fig:localmod}). If $p$ and $p'$ are separated in $\M'$ but not in $\M$ it means that there exists $p''$ such that $p < p'' < p'$ and $p'' \not \in P$. But then the original \vgraph\ $\ver$ had an edge between $(p, \vect{p}_{-i})$ and $(p'', \vect{p}''_{-i})$ and an edge between $(p'', \vect{p}''_{-i})$ and $(p', \vect{p}'_{-i})$ for all $ \vect{p}_{-i}, \vect{p}'_{-i}, \vect{p}''_{-i} \in D_{-i}^{(u)}$. Since $\M$ is 2CMON then the labels for $p$ at $u$ are larger than the labels for $p''$ at $u$ that are larger than the labels for $p'$ at $u$, thus implying that $\M'$ is 2CMON too. 
 
 Suppose now that $\M'$ is not OSP. Then there is a negative cycle in some graph $\verp$. Let us consider the negative cycle $C$ in $\verp$ with fewer edges. Since this negative cycle does not exist in \vgraph\ $\ver$ of the original mechanism, then $C$ must use at least one edge $(\x,\y)$ that exists in $\verp$ but not in $\ver$ (this is only possible if both $x_i$ and $y_i$ were in $P$, for $P \in \{L, R\}$). We will show that there exists another negative cycle that instead of $(\x,\y)$ uses edges already existing in $\ver$. By a repeated application of this argument we finish with a negative cycle that only uses edges defined by the queries of the original mechanism -- a contradiction.
 
 In the following, we will suppose without loss of generality that $x_i, y_i \in R$ and let $j$ be the index such that $x_i \in R_j$.
 We distinguish two cases. 
 \begin{description}[topsep=0pt,style=unboxed,leftmargin=2ex]
  \item[Case $f_i(\x) \leq f_i(\y)$.] Observe that, since $\M'$ satisfies 2CMON, then it must be the case that $x_i > y_i$ and thus $y_i \in R_k$ for some $k < j$.
  
  Let $\z$ be any profile compatible with $u$ such that $z_i \in L_{j-1}$. Observe that $x_i > z_i > y_i$. Moreover, edges $(\x, \z)$ and $(\z,\y)$ belong to both $\ver$ and $\verp$. Then the claim follows since, by Lemma~\ref{lem:bypass}, we can replace edge $(\x,\y)$ with $\x \rightarrow \z \rightarrow \y$.
  
  \item[Case $f_i(\x) > f_i(\y)$.] Observe that, since $\M'$ satisfies 2CMON, then it must be the case that $x_i < y_i$ and thus $y_i \in R_k$ for some $k > j$.
  
  Since $C$ is a cycle, it must contain a path $\x \rightarrow \y^1 \rightarrow \ldots \rightarrow \y^\ell \rightarrow \x'$ for some $\ell \geq 1$ such that $\y^1 = \y$ (i.e., we are considering the path that starts with the edge $(\x, \y)$ on which we are focusing), $\x'$ is such that $x'_i \in R_p$ for some $p \leq j$ (note that we allow that $x'_i = x_i$ and $\x' = \x$), and $y_i^s \notin \bigcup_{q \leq j} R_q$ for every $s \in [\ell]$ (i.e., $x'_i$ is the first profile along the path that belongs to $R_j$ or a subset of $R$ containing smaller types).
  
  We first prove that $y_i^s$ does not belong to $\bigcup_{q \leq j+1} L_{q}$ for every $s \in [\ell]$. Suppose instead that this is not the case and consider the first $s \in [\ell]$ such that $y_i^s \in L_{q}$ for some $q \leq j+1$. Consider the path $\x \rightarrow \y^1 \rightarrow \ldots \rightarrow \y^{s-1} \rightarrow \y^{s}$. Note that $s > 1$, since by definition $\y^1 = \y$ and thus $y^1_i \notin L$. Now, observe that for every $r \in [s-1]$ either $y_i^r \in R_t$ for some $t > j$ or $y_i^r \in L_t$ for some $t > j+1$, and thus $y_i^t$ does not belong to either $L_q$ or $R_j$. Hence, we have that  for every $r \in [s-1]$ the edge $(\y^s,\y^r)$ must exist in $\verp$, and, hence, since $y_i^s < y_i^r$, 2CMON implies that $f_i(\y^s) \geq f_i(\y^r)$.
  Moreover, we must have that $x_i < y_i^{r}$ for every $r \in [s-1]$. Hence $(\x, \y^1, \ldots, \y^{s-1}, \y^{s})$ is a good sequence of profiles. Furthermore, since the edge $(\x,\y^s)$ exists, it follows from Corollary~\ref{cor:angle_redux} that the cycle $C'$ obtained from $C$ by replacing the path $\x \rightarrow \y_1 \rightarrow \ldots \rightarrow \y^s$ with the edge $(\x,\y^s)$ exists in $\verp$ and is negative. However this is in contradiction with our choice of $C$ as the negative cycle with fewer edges.
  
  Now, let $\z$ be any profile compatible with $u$ such that $z_i \in L_{j+1}$. Observe that $x_i < z_i$ and $x'_i < z_i$. Moreover, for every $q \in [\ell]$ we have that $y_i^q \in R_k$ for some $k > j$ or $y_i^q \in L_k$ for some $k > j+1$, and thus $y_i^q$ does not belong to either $L_{j+1}$ or $R_j$. Hence, it follows that $x_i < z_i < y_i^q$ for every $q \in [\ell]$. Moreover, edge $(\z,\y^q)$ exists in $\verp$ for every $q \in [\ell]$, and hence 2CMON yields $f_i(\z) \geq f_i(\y^q)$, since $z_i < y_i^q$. Similarly, edge $(\z, \x')$ exists in $\verp$, and, therefore, since $z_i > x'_i$, 2CMON implies that $f_i(\x') \geq f_i(\z)$. Finally, observe that both $(\x,\z)$ and $(\z, \x')$ exist in $\verp$. Then, we can conclude that the sequence of profiles $(\x, \y^1, \ldots, \y^\ell, \x')$ satisfies the conditions of Corollary~\ref{cor:angle_replace}, from which we achieve that the cycle $C'$ obtained from $C$ by replacing the path $\x \rightarrow \y_1 \rightarrow \ldots \rightarrow \y^\ell \rightarrow \x'$ with the path $\x \rightarrow \z \rightarrow \x'$ exists in $\verp$ and is negative. If $\ell \geq 1$ this contradicts our choice of $C$. Otherwise the claim follows since both edges $(\x,\z)$ and $(\z,\x')$ exist also in $\ver$.\qedhere
 \end{description}
\end{proof}

\section{Machine Scheduling}
In this section, we will give applications of our characterization to close the gap on the approximation that OSP mechanisms can guarantee for scheduling related machines. Moreover, these applications show that cycle of length four are often sufficient to prove the limits of OSP mechanisms.

\paragraph{Five (or more) types in the domain.}
We begin with the following result.
\begin{theorem}\label{thm:lbn}
	There is no OSP mechanism with approximation guarantee better than $n$ for the scheduling related machines problem. This is true even if all the $n$ agents have type from the same domain of size five.
\end{theorem}

To prove the theorem, consider the following setting: $m=n$, and each agent has the following five types in the domain: $B, L > n^2B, M > n^2L, H > n^2M, T > n^2H$.
Suppose that there is an OSP mechanism $\M$ that 
returns a solution with approximation $\rho < n$, and let us consider the path $P$ of the implementation tree $\T$ of $\M$ compatible with the type of all agents being $H$.

Note that since the mechanism is OSP, then it satisfies 2CMON. We next prove some properties about the structure of $\M$ along $P$ based only on its approximation ratio and the 2CMON property. Subsequently, we show how we can provide an even stronger constraint on the structure of the mechanism by considering longer cycles. This structural property will imply the desired result.

We say that an agent is \emph{asked to remove $B$ from her domain} if she receives one of three queries: either (i) a greedy query about $B$ (to which the agent answers negatively along $P$), or (ii) a query asking to split the domain around $M$, i.e., $\{B,L\}$ and $\{H,T\}$ with $M$ in either part (to which the agent declares that her type is in the part including $H$ along $P$), or (iii) a reverse greedy query about $H$ (to which the agent answers positively along $P$). 
We first provide a simple observation, that follows from the necessity that, by approximation guarantee, the mechanism along $P$ must return different outcomes when an agent has type $B$ or $H$.
\begin{observation}[Removing $B$]
	\label{obs:notB}
	Along $P$, all agents must be asked to remove $B$ from their domain.
\end{observation}
\begin{proof}
	Suppose for the sake of contradiction that there is an agent $i^*$ who is not asked to remove $B$ from her domain. Let $\x$ be the profile such that $x_i = H$ for every machine $i$, and  $\y$ be the profile wherein $y_{i^*} = B$ and $y_i = H$ for all $i \neq i^*$. Since along the path $P$ leading to $\x$, agent $i^*$ is never asked to remove $B$ then she does not separate $B$ from $H$. Thus, it must be the case that the outcome that the mechanism assigns to $i^*$ is the same in both $\x$ and $\y$. By approximation, the mechanism must assign all jobs to $i^*$. 
	However, this means that the mechanism has makespan  $mH = nH$ for the type profile $\x$ 
	whereas the optimum 
	would assign each job to a different machine with makespan $H$, from which it follows that the approximation of the mechanism must be $n$, a contradiction.
\end{proof}

We say that an agent $i$ is \emph{asked to separate $L$ and $H$} if either 
	(i) $i$ receives a reverse greedy query about type $H$ or a query asking to split the domain around $M$; or, (ii) $i$ receives a greedy query about type $L$. \label{item:low}
Note that an agent $i$ who separates $L$ and $H$ at some node $u$ is either separating also $B$ and $H$ at the same node (as it is the case when condition (i) occurs), or, since without loss of generality the queries are ordered, she has separated $B$ and $H$ at some node preceding $u$ along $P$.

Let now $u^*$ be the last node in $\T$ along $P$  where agent $i^*=i(u^*)$ is asked to remove $B$ from her domain. 
This means that for every node $u \in \T$ preceding $u^*$ in $P$, it holds that $B \in D_{i^*}^{(u)}$. We say that an agent $i$ is \emph{excluded} if she separates $L$ and $H$ at a node $u_i$ preceding $u^*$ in $P$. We also denote the child $v_i$ of $u_i$ along $P$ (i.e., the one child of $u_i$ corresponding to $i$ answering the query in a way that is compatible with her type being $H$) as the \emph{excluding} node of $i$.

Let $Z$ be the set of excluded agents. We next prove that along the path $P$, the mechanism $\M$ cannot assign a positive job load to any agent in $Z$, since at the time they separate $L$ from $H$ it is still possible that $i^*$ has type $B \ll L$, and thus, by approximation, $i^*$ must receive all the jobs.
\begin{observation}[No job to agents separating $L$ and $H$]
	\label{obs:excluded}
	For every $i \in Z$, the mechanism $\M$ must assign outcome zero to $i$ in each profile compatible with its excluding node $v_i$.
\end{observation}
\begin{proof}
	Suppose by contradiction that there is one profile $\x$ compatible with $v_i$ such that $f_i(\x) > 0$. Note that the compatibility with $v_i$ implies that $x_i \geq H$.
	Let $\y$ be the type profile such that $y_i = L$, $y_{i^*} = B$ and $y_j = H$ for every $j \neq i, i^*$. Since $u_i$ precedes $u^*$, both $\x$ and $\y$ are compatible with node $u_i$.
	Since $\M$ satisfies 2CMON and as $\x$ and $\y$ are compatible with $u_i$ and separated at $u_i$, it must be the case that $0 < f_{i}(\x) \leq f_{i}(\y)$. Hence, $\M$, when the type profile is $\y$, must assign at least one job to $i \neq i^*$. Since $L > n^2 B$, then the makespan of this solution is at least $L$, while the optimal outcome assigns all jobs to $i^*$ for a makespan of $nB$. Hence, $\M$ has approximation larger than $n$, which contradicts our hypothesis.
\end{proof}
Observe that $i^*$ cannot be excluded, since, by construction she cannot separate $L$ and $H$ before separating $B$ and $H$. Hence $|Z| < n$. We next show that the set of excluded agents is non-empty, for otherwise all nodes must receive the same outcome when their type is $L$ and $H$, and this can break the approximation of the algorithm.
\begin{observation}[Existence of $L$ and $H$ separations]
	\label{obs:zeta}
	It holds that  $|Z| \geq 1$.
\end{observation}
\begin{proof}
	Suppose instead that $|Z| = 0$, and thus each agent $i \neq i^*$ along $P$ separates $L$ and $H$ at a node $u_i$ that follows $u^*$ in $P$. Let $u'$ be the first node at which an agent is asked to remove $B$ from her domain along $P$, and let $i'$ be the agent queried at this node. By definition, $u'$ is an ancestor of $u^*$ and then $D_i^{(u')} \supseteq \{B, H\}$ for each agent $i$. Moreover, since $u^*$ precedes $u_{i}$, it must be the case that $D_i^{(u')} \supseteq \{B, L, M, H\}$ for each $i$.
	
	Let $\x$ be the profile such that $x_i = B$ for all $i$, and let $\y$ be such that $y_{i'} = L$ and $y_i = H$ for all $i \neq i'$. Note that both $\x$ and $\y$ are compatible with node $u'$. Then, by 2CMON, $f_{i'}(\x) \geq f_{i'}(\y)$.
	
	If $f_{i'}(\x) = m$ (i.e., $i'$ receives all jobs), then the outcome returned by $\M$ has makespan $nB$, while the optimal outcome, which assigns each job to a different agent, has makespan $B$. Hence, $\M$ has approximation $n$, which contradicts our hypothesis.
	
	If $f_{i'}(\x) < m$, then $f_{i'}(\y) < m$, and thus there is an agent $i \neq i'$ that is assigned at least one job by $\M$ when the type profile is $\y$. Since $H > n^2 L$, then the makespan of this solution is at least $H$, while the optimal outcome assigns all the jobs to $i^*$ for a makespan of $nL$. Hence, $\M$ has approximation larger than $n$, which contradicts our hypothesis.
\end{proof}

We say that an agent $i$ is \emph{asked to separate $L$ from $T$} if $i$ either receives a reverse greedy query about $T$, or a query in which she is asked to split the domain around $M$, 
or a greedy query about $L$. Let $\overline{Z} = [n] \setminus Z$ be the set of all agents that have not separated $L$ and $H$ before node $u^*$.
As for $L$ and $H$ separations, we can prove that all agents in $\overline{Z}$ must separate $L$ from $T$ along $P$.
\begin{observation}[Existence of $L$ and $T$ separations]
	\label{obs:LvsT}
	Along $P$, all agents in $\overline{Z}$ must be asked to separate $L$ from $T$.
\end{observation}
\begin{proof}
	Suppose instead that there is an agent $j \in \overline{Z}$ who is not asked to separate $L$ from $T$. Since $\M$ only makes ordered queries, it is not possible to separate $H$ from $L$ without separating $L$ from $T$. Hence, it must be the case that $D_j^{(u)} \supseteq \{L, M, H, T\}$ for every node $u \in P$.
	
	Consider then the profile $\x$ such that $x_j = T$ and $x_i = H$ for every $i \neq j$, and the profile $\y$ such that $y_{j} = L$ and $y_i = H$ for all $i \neq j$. Since $D_j^{(u)} \supseteq \{L, M, H, T\}$ for every node $u \in P$, it must be the case that the mechanism allocates the same outomce to $j$ in both type profiles $\x$ and $\y$. Since, by approximation guarantee, the mechanism assigns all jobs to $j$ in $\y$ 
	then  $\M$ returns a solution for the $\x$ that has makespan $nT$, whereas 
	the optimum 
	would assign all the jobs to machines different from $j$, with a makespan at most $nH$. It follows that the approximation of the mechanism must be larger than $n$, a contradiction.
\end{proof}
For each $i \in \overline{Z}$, let $\ell_i$ be the node of $\T$ in $P$ at which $i$ is asked to separate $L$ from $T$. By Observation~\ref{obs:LvsT}, $\ell_i$ is well defined for every $i \in \overline{Z}$.
We will prove a relationship between $u^*$ and $\ell_i$  along $P$. To this aim, we first need the following three observations characterizing the output of the mechanism for certain profiles reachable in $P$.
\begin{observation}[Full load to $B$]
	\label{obs:all}
	Let $u$ be a node along $P$ at which an agent $i$ is asked to remove $B$ from her domain. Then in every profile $\x$ compatible with $u$ such that $x_i = B$, $f_i(\x) = m$.
\end{observation}
\begin{proof}
	Suppose that there is one profile $\x$ compatible with $u$ such that $x_{i} = B$ but $f_i(\x) < m$. Let $\y$ be such that $y_{i} = L$ and $y_j = H$ for all $j \neq i$.  By ordered queries, agent $i$ cannot separate $L$ from $H$ before separating $B$ from $H$ along $P$. Thus, since $B$ is still in the domain, $L$ belongs to the current domain as well and then $\y$ is compatible with node $u$.
	
	Since $\M$ satisfies 2CMON, it must be the case that $m > f_{i}(\x) \geq f_{i}(\y)$. Hence, $\M$, when the type profile is $\y$, must assign at least one job to a machine $j \neq i$. Since $H > n^2 L$, then the makespan of this solution is at least $H$, while the optimal outcome, that assigns all the jobs to $i$, has makespan $nL$. Hence, $\M$ has approximation guarantee larger than $n$, which contradicts our hypothesis.
\end{proof}

\begin{observation}[No job to agents separating $L$ and $T$/1]
	\label{obs:zero}
	Let $u$ be a node along $P$ at which an agent $i$ is asked to separate $L$ from $T$, and suppose that there is $i^*$ such that $B \in D_{i^*}^{(u)}$. Then in every profile $\x$ compatible with $u$ such that $x_i = T$, $f_i(\x) = 0$.
\end{observation}
\begin{proof}
	Suppose instead that there is one profile $\x$ compatible with $u$ such that $x_i = T$ and $f_i(\x) > 0$. Let $\y$ be the type profile such that $y_i = L$, $y_{i^*} = B$ and $y_j = H$ for every $j \neq i, i^*$. Observe that $\y$ also is compatible with node $u$ since $u$ is an $(L,T)$-separating node.
	
	Since $\M$ satisfies 2CMON, it must be the case that $0 < f_{i}(\x) \leq f_{i}(\y)$. Hence, $\M$   must assign at least one job to $i$ when the type profile is $\y$. Since $L > n^2 B$, then the makespan of this solution is at least $L$, while the optimal outcome, which assigns all the jobs to $i^*$, has makespan $nB$. Hence, $\M$ has approximation guarantee larger than $n$, contradicting our hypothesis.
\end{proof}

\begin{observation}[No job to agents separating $L$ and $T$/2]
	\label{obs:zero2}
	Let $u$ be a node along $P$ at which an agent $i$ is asked a reverse greedy query about $T$, and suppose that there is $i^*$ such that $L \in D_{i^*}^{(u)}$. Then in every profile $\x$ compatible with $u$ such that $x_i = T$, $f_i(\x) = 0$.
\end{observation}
\begin{proof}
	Suppose instead that there is one profile $\x$ compatible with $u$ such that $x_i = T$ and $f_i(\x) > 0$. Let $\y$ be the type profile such that $y_i = H$, $y_{i^*} = L$ and $y_j = H$ for every $j \neq i, i^*$. Observe that $\y$ also is compatible with node $u$ since $u$ is an $(H,T)$-separating node.
	
	Since $\M$ satisfies 2CMON, it must be the case that $0 < f_{i}(\x) \leq f_{i}(\y)$. Hence, $\M$must assign at least one job to $i$ in input $\y$. Since $H > n^2 L$, then the makespan of this solution is at least $H$, while the optimal outcome -- which assigns all the jobs to $i^*$ has makespan $nL$. 
	This contradicts our hypothesis.
\end{proof}
Now we are ready to prove the desired relation between $u^*$ and $\ell_i$ for $i \in \overline{Z}$. We remark that for this claim two cycle turn out to be insufficient, and we need to resort to cycles of longer length.
\begin{lemma}
	\label{lem:4cycle}
	For every $i \in \overline{Z}$, $u^*$ precedes $\ell_i$  along $P$.
\end{lemma}
\begin{proof}
	Suppose that this is not the case, and there is $j \in \overline{Z}$ such that $\ell_j$ precedes $u^*$.
	We let $j^*$ be the agent with the highest $\ell_j$ in the implementation tree $\T$, that is, the $j \in \overline{Z}$ such that $\ell_{j^*}$ precedes $\ell_j$ for every $j\in \overline{Z}$ with $j \neq j^*$. This implies that $D_j^{(\ell_{j^*})} \supseteq \{L, M,H, T\}$ for every $j \in \overline{Z}$. Moreover, $D_{i^*}^{(\ell_{j^*})} \supseteq \{B, H\}$.
	We distinguish two cases based on the identity of $j^*$.
	\begin{description}[topsep=0pt,style=unboxed,leftmargin=2ex]
		\item[Case $j^* \neq i^*$.] Let $w^*$ be the node at which $j^*$ separates $B$ from $T$. Note that either $w^*$ precedes $\ell_{j^*}$ (if at $w^*$ agent $j^*$ receives a greedy query about $B$), or $w^* = \ell_{j^*}$ (since, if $B$ and $T$ have not been separated before, they will be surely separated when $L$ and $T$ are separated). Hence, since $D_{j^*}^{(\ell_{j^*})} \supseteq \{L, M,H, T\}$ and $B$ is still available at $w^*$, we have that $D_{j^*}^{(w^*)} = \{B, L, M, H, T\}$.
		
		According to Observation~\ref{obs:all}, every profile $\s$ compatible with $w^*$ such that $s_{j^*} = B$ must have $f_i(\s) = m$. Similarly, according to Observation~\ref{obs:zero}, every profile $\t$ compatible with $\ell_{j^*}$ (and thus with $w^*$) such that $t_{j^*} = T$ must have $f_i(\t) = 0$.
		
		Consider now profiles $\x$ and $\y$ such that $x_{j^*} = L$, $x_{i^*} = B$, and $x_i = H$ for every $i \neq i^*, j^*$, and $y_{j^*} = H$, $y_j = T$ for every $j \in \overline{Z}$ with $j \neq i$ and $y_j = H$ for every $j \in Z$. Note that both $\x$ and $\y$ are compatible with $\ell_{j^*}$ and thus with $w^*$.
		
		Observe that $\M$ must assign outcome $0$ to $j^*$ on input the profile $\x$: indeed, for this type profile the optimal outcome would be to assign all jobs to $i^*$ with makespan $nB$, and any solution that assigns a job to a machine different from $i^*$ would have approximation worse than $n$. Moreover, $\M$ must assign outcome $m$ to $j^*$ on input profile $\y$. Suppose indeed that this is not the case and there is at least a job assigned to a machine $j \neq j^*$: note that $j$ cannot be an excluded agent, since $\y$ is compatible with its excluding node $v_j$ (recall that $y_j = H$), and thus, by Observation~\ref{obs:excluded}, $f_j(\y) = 0$; hence, it must be the case that $j \in \overline{Z}$, and thus, since $T > n^2H$, the makespan of the mechanism is $T$. However, for the type profile $\y$ the optimal outcome would be to assign all jobs to those machines that have type $H$, with makespan at most $nH$. Then, $\M$ has approximation worse than $n$, that is contradiction.
		
		Hence, $x_{j^*}$ and $y_{j^*}$ are antimonotone types witnessed by $\x$ and $\y$. Therefore, by 2CMON, there cannot be an edge between them in the \vgraph\ of $j^*$. However, since they are both still available at $\ell^{j^*}$, this means they have been separated after this node, and that $\s$ and $\t$ are anchors for this nodes. That is, in the \vgraph\ of $j^*$ there is the cycle $C = \y \rightarrow \t \rightarrow \x \rightarrow \s \rightarrow \y$ that costs $(L-H)m < 0$, which contradicts the fact that $\M$ is OSP.
		
		\item[Case $j^* = i^*$.] In this case, since, as showed above, $D_{i^*}^{(\ell_{i^*})} = D_{j^*}^{(\ell_{j^*})} \supseteq \{L, M, H, T\}$ and $D_{i^*}^{(\ell_{i^*})} = D_{i^*}{(\ell_{j^*})} \supseteq \{B, H\}$, we have that $D_{i^*}^{(\ell_{i^*})}= \{B, L, M, H, T\}$. Moreover, since $\ell_{i^*}$ precedes $u^*$, and thus at $\ell_{i^*}$ we are separating $L$ from $T$, but not $B$ from $H$, it must be that at this node we are making a reverse greedy query about type $T$ (and not a split query).
		
		As above, Observation~\ref{obs:all} implies that every profile $\s$ compatible with $u^*$ (and thus with $\ell_{i^*}$) such that $s_{i^*} = B$ must have $f_i(\s) = m$. Similarly, according to Observation~\ref{obs:zero2}, every profile $\t$ compatible with $\ell_{i^*}$ such that $t_{i^*} = T$ must have $f_i(\t) = 0$.
		
		Consider now profiles $\x$ and $\y$ as follows. In $\x$, we have $x_{i^*} = M$, $x_j = L$ for every $j \in \overline{Z}$, with $j \neq i^*$, and $x_j = H$ for every $j \in Z$. In $\y$, we have $y_{i^*} = H$, $y_j = T$ for every $j \in \overline{Z}$, with $j \neq i^*$, and $y_j = H$ for every $j \in Z$. Note that both $\x$ and $\y$ are compatible with $\ell_{i^*}$. 
		
		Observe that $\M$ must assign outcome $0$ to $i^*$ on input the profile $\x$: indeed, for $\x$ the optimal outcome would be to uniformly allocate all the jobs to agents $j \in \overline{Z}$, with $j \neq i^*$, for a makespan of at most $nL$, whereas any solution that assigns a job to a different machine would have approximation worse than $n$.
		Moreover, $\M$ must assign outcome $m$ to $i^*$ on input profile $\y$. Suppose indeed that this is not the case and there is at least a job assigned to a machine $j \neq i^*$. note that $j$ cannot be an excluded node. In fact, since $\y$ is compatible with its excluding node $v_j$ (recall that $y_j = H$), Observation~\ref{obs:excluded} implies $f_j(\y) = 0$. Hence, it must be the case that $j \in \overline{Z}$, and thus, since $T > n^2H$, the makespan of the mechanism is $T$. However, for the type profile $\y$ the optimal outcome would be to uniformly assign all the jobs to those machines that have type $H$, with makespan at most $nH$. Then, $\M$ has approximation worse than $n$, that is contradiction.
		
		Hence, $x_{j^*}$ and $y_{j^*}$ are antimonotone types witnessed by $\x$ and $\y$. Therefore, by 2CMON, there cannot be an edge between them in the \vgraph\ of $j^*$. However, since they are both still available at $\ell^{j^*}$, this means they have been separated after this node, and that $\s$ and $\t$ are anchors for this nodes. That is, in the \vgraph\ of $i^*$ there is the cycle $C = \y \rightarrow \t \rightarrow \x \rightarrow \s \rightarrow \y$ of weight $(L-H)m < 0$, which contradicts the fact that $\M$ is OSP.\qedhere
	\end{description}
\end{proof}

\begin{proof}[Proof of Theorem \ref{thm:lbn}]
By Lemma~\ref{lem:4cycle}, we then have that $D_{i^*}^{(u^*)} = \{B, L, M, H, T\}$ and $D_i^{(u^*)} = \{L, M, H, T\}$ for every $i \in \overline{Z}$, with $i \neq i^*$.
Consider then the profile $\x$ such that $x_i = H$ for $i \in Z$, and $x_i = T$ for $i \in \overline{Z}$. Note that $\x$ is compatible with $u^*$. Observe that the mechanism $\M$ assigns outcome $0$ to every machine $i \in Z$, since $\x$ is compatible with $v_i$ (because $x_i = H$). Hence, the mechanism must assign jobs only to machines in $\overline{Z}$, with a makespan that is at least $T$. However, the optimal allocation would assign all the jobs to the machines in $Z$. Since, by Observation~\ref{obs:zeta}, $|Z| \geq 1$, it follows that the optimal makespan is at most $nH$. Hence, since $T > n^2H$, we have that $\M$ has an approximation worse that $n$, that is a contradiction.
\end{proof}

\paragraph{Four types.}
We now move to domains of size four and prove the following result.

\begin{theorem}\label{thm:lb4}
	There is no OSP mechanism with approximation guarantee better than {$n/2 + 1$} for the scheduling related machines problem when all the $n$ agents have type from the same domain of size four.
\end{theorem}

Consider the following setting: {$n$ is even, $m=cn$, with $c=\frac{n}{2} + 1$} and each agent has the following four types in the domain: {$B, L > mnB, H > mnL, T > mnH$}. By inspection, it is not hard to see that all the observations above continue to hold whereas Lemma \ref{lem:4cycle} ceases to be true. However, we can prove a slightly weaker result by using the case $j^* \neq i^*$ in the proof of Lemma \ref{lem:4cycle}.
\begin{lemma}
	\label{lem:4cycle4}
	For every $i \in \overline{Z}$ {with $i \neq i^*$}, $u^*$ precedes $\ell_i$  along $P$.
\end{lemma}
Moreover, we can strengthen the lower bound on the size of $Z$ as follows.
	\begin{lemma}
		\label{lem:sizeZ}
		It holds $|Z| > \frac{n}{2} - 1$.
	\end{lemma}
	\begin{proof}
		Suppose instead that $|Z| \leq \frac{n}{2} - 1$, and hence $|\overline{Z}| \geq \frac{n}{2}+1$. Let $u$ be the first node along $P$ in which an agent $i \in \overline{Z}$ removes $B$ from their domain. By definition of $u$, we must have that $B \in D_j^{(u)}$ for every $j \in \overline{Z}$. Consider then the profile $\x$ such that $x_j = B$ for every $j \in \overline{Z}$ and $x_j = H$ for every $j \in Z$. Note that $\x$ is compatible with $u$. Thus, by Observation~\ref{obs:all}, we have that $f_i(\x) = m$, and thus the mechanism has makespan $mB$. However, the optimal mechanism on this instance fairly split jobs among all machines in $\overline{Z}$, with makespan $\left\lceil\frac{m}{|\overline{Z}|}\right\rceil B \leq \left\lceil\frac{2m}{n+2}\right\rceil B = \left\lceil\frac{n(n+2)}{n+2}\right\rceil B = nB$. Hence, the mechanism has approximation $\frac{m}{n} = c = \frac{n}{2} + 1$, a contradiction.
	\end{proof}

\begin{proof}[Proof of Theorem \ref{thm:lb4}]
	By Lemma~\ref{lem:4cycle4}, we have that $D_{i^*}^{(u^*)} { \supseteq \{B, L, H\}}$ and $D_i^{(u^*)} = \{L, H, T\}$ for every $i \in \overline{Z}${ , with $i \neq i^*$}.
	Let $\x$ be the profile such that {  $x_i^* = H$,} $x_i = H$ for $i \in Z$, and $x_i = T$ for $i \in \overline{Z}$; $\x$ is compatible with $u^*$. The mechanism $\M$ assigns outcome $0$ to every machine $i \in Z$, since $\x$ is compatible with $v_i$ (because $x_i = H$). Hence, the mechanism must assign jobs only to machines in $\overline{Z}$, with a makespan that is at least {  $mH$}. However, the optimal allocation would fairly split the jobs among the machines in ${  Z \cup \{i^*\}}$. By Lemma~{ \ref{lem:sizeZ}, $|Z \cup \{i^*\}| > n/2$. Since $n$ is even, then $|Z \cup \{i^*\}| \geq \frac{n}{2} + 1$}. It follows that the optimal makespan is { $\left\lceil\frac{m}{|Z \cup \{i^*\}|}\right\rceil B \leq \left\lceil\frac{2m}{n+2}\right\rceil B = \left\lceil\frac{n(n+2)}{n+2}\right\rceil B = nB$}. Hence, the mechanism has approximation at least {$\frac{m}{n} = c = \frac{n}{2} + 1$}, a contradiction.
\end{proof}

It is not hard to see that a stronger result can be proved if one considers heterogeneous domains. Specifically, for every OSP mechanism whose query order is independent from the domain of agents, there is an instance on which it cannot achieve an approximation better than $n$: indeed, it is sufficient to take $H_{i^*} \gg H$ in the proof of Theorem~\ref{thm:lb4}. For this reason, we next focus only on the case of homogeneous domains.

\paragraph{Mechanism for four types.}
\newcommand{\Greedy}{\texttt{Greedy}}
\newcommand{\Mmany}{\M_{4}}
We now introduce mechanism $\Mmany$.
The mechanism adopts a simple routine \Greedy$(t, \pi)$, that consists in asking all the agents (belonging to a given subset) -- in round robin fashion according to the order $\pi$ -- if their type is the smallest not yet removed type in their domain up to type $t$, and assigning all jobs to the first machine answering yes to one of these queries. When we omit the parameter $\pi$, every order may be used. We say that \Greedy$(t, \pi)$ fails if no agent is found with type at most $t$.

The mechanism essentially uses a reverse greedy (a.k.a., descending) phase to find the $n/2$ machines with largest type to which it will assign outcome $0$. Subsequently, the mechanism uses a greedy (a.k.a., ascending) phase over the remaining $n/2$ machines using the \Greedy~ routine. As described above, however, for an OSP mechanism these two phases should not be combined, that is, we need to avoid querying the same agent both in the descending and ascending phase until this process would create two antimonotone profiles with two pivots (that happens to be until there is at least one agent who has revealed to have type at most $L$). However, to keep the approximation of the mechanism bounded we should mix phases for some agents, for otherwise it may occur that the ones that have not yet been queried in the descending phase are the ones with the worst type in the domain. Hence, we need to select a special agent that will be queried in both phases. However, in order to avoid a negative cycle for this agent we need to play with the timing in which this agent is queried during the ascending phase: specifically, we need to force her to be the last to be queried about type $B$ and the first to be queried about type $L$. The mechanism is given in Algorithm \ref{descending}.

\begin{algorithm}[htbp]
\DontPrintSemicolon
\small
Let $i^* = 1$ (special agent), $A = [n]$ (alive machines), $t = T$ (minimum largest type of alive machines)\;
\tcc*[f]{Descending Phase}\;
\While{receiving a yes answer to previous queries (if any) and $|A| > \lceil n/2 \rceil$}{
Ask $i^*$ if her type is $T$\;
\lIf{yes}{Remove $i^*$ from $A$ and set $i^* = i^*+1$}
}
\lIf{$|A| > \lceil n/2 \rceil$}{Set $t = H$}
Set $i = i^* + 1$\;
\While{receiving an answer $\{H, T\}$ to previous split queries (if any) and $|A| > \lceil n/2 \rceil$}{
Ask $i$ if her type is in $\{B, L\}$ or in $\{H, T\}$\;
\lIf{answer is $\{H, T\}$}{Remove $i$ from $A$ and set $i = i+1$}
}
\lIf{$|A| > \lceil n/2 \rceil$}{Set $t = L$ and $i^*$ to be the last queried machine (whose answer has been $\{B, L\}$)}
Let $i$ be the machine in $A$ with the lowest id\;
\While{$i \leq n$ and $|A| > \lceil n/2 \rceil$}{
Ask $i$ if her type is in $\{B, L\}$ or in $\{H, T\}$\;
\lIf{answer is $\{H, T\}$}{Remove $i$ from $A$ and set $i$ to the next machine in $A$}
}
Let $i$ be the machine in $A$ with the lowest id\;
\While{$i \leq n$ and $|A| > \lceil n/2 \rceil$}{
Ask $i$ if her type is $L$\;
\lIf{yes}{Remove $i$ from $A$ and set $i$ to the next machine in $A$}
}
\lIf{$|A| > \lceil n/2 \rceil$}{Set $t = B$}

\tcc*[f]{Ascending Phase}\;
\lIf{$|A| > \lceil n/2 \rceil$}{Evenly split the $m$ jobs over machines in $A$ \label{line:outcome1}}
\ElseIf{$t < H$}{
Run \Greedy$(B)$ over machines in  $A$ \label{line:outcome2}\;
\If{it fails}{
Ask all machines in $A$ different from $i^*$ if her type is $L$ until one answers yes \label{line:notall2}\;
\lIf{there is a machine $i$ that answered yes}{Assign $\lceil m/2\rceil$ jobs to $i^*$ and $\lfloor m/2\rfloor$ jobs to $i$ \label{line:outcome3}}
\lElse{Assign $m$ jobs to $i^*$ \label{line:outcome3.5}}
}
}\Else
{Run \Greedy$(B, \pi)$ over machines in $A$ for some order $\pi$ such that $i^*$ is ranked last \label{line:outcome4}\;
\If{it fails}
{Run \Greedy$(L, \pi)$ over machines in $A$ for some order $\pi$ such that $i^*$ is ranked first \label{line:outcome5}\;
\If{it fails}{
Ask all machines in $A$ different from $i^*$ if her type is $H$ until one answers yes \label{line:notall}\;
\lIf{there is a machine $i$ that answered yes}{Assign $\lceil m/2\rceil$ jobs to $i^*$ and $\lfloor m/2\rfloor$ jobs to $i$ \label{line:outcome6}}
\lElseIf{$t=H$}{Assign $m$ jobs to $i^*$ \label{line:outcome7}}
\lElse{Assign $\lceil m/2\rceil$ jobs to $i^*$ and $\lfloor m/2\rfloor$ job to some machine $i \in A$ with $i \neq i^*$ \label{line:outcome8}}
}}}
 \caption{Mechanisms $\Mmany$}
 \label{descending}
\end{algorithm}

\begin{theorem}
	\label{prop:algo2mon}
	Mechanism $\Mmany$ is OSP when each agent has type in $D = \{B, L, H, T\}$.
\end{theorem}
\begin{proof}
	We first prove that $\Mmany$ satisfies OSP 2CMON. Indeed, whenever a machine receives a query in the descending phase, she gets no job if she declares a large type, and she cannot receive fewer jobs with a smaller type. Similarly, whenever a machine receives a query in the ascending phase, she gets all the jobs if she declare a small type, and she cannot receive more jobs with a larger type. The only exception is for queries at Line~\ref{line:notall2} and Line \ref{line:notall}: however, in both cases, since we are in the ascending phase, the machine receives $\lfloor m/2\rfloor$ jobs if her type is small and at most $\lfloor m/2\rfloor$ otherwise.

	Next we prove that the mechanism is three-way greedy, that is, for every agent $i$ it is not possible to find two antimonotone profiles with two pivots. This is sufficient by Theorem~\ref{thm:main}.
	Note that with four types, this may only occur if $i$ has been queried first about $T$ ($B$), answers negatively, then about $B$ ($T$, respectively), answers negatively, and only later she separates $L$ and $H$. It is immediate to check that the only agent for which this structure exists in $\Mmany$ is $i^*$ when $t \geq H$. However, as we will show below, for this agent the outcomes associated to types $L$ and $H$ that can arise after $i^*$ has been queried both about type $T$ and type $B$, imply that $L$ and $H$ are not antimonotone.

	To this aim, observe that, after these queries, the domain of $i^*$ contains only $L$ and $H$. Moreover, since $i^*$ is the last one to be queried about type $B$, it must be the case that the type of each alive agent must be at least $L$, and the type of each non-alive machine is at least $H$.
	Hence, since $i^*$ is the first agent to be queried about type $L$ by the routine \Greedy, it must be the case that, whenever the structure described above occurs, $i^*$ receives outcome $m$ for type $L$. Since the mechanism cannot assign a larger outcome for every profile in which the type of $i^*$ is $H$, we have that $L$ and $H$ cannot be antimonotone.
\end{proof}

\begin{theorem}
\label{cor:algoapprox}
Mechanism $\Mmany$ is $\left(\frac{n}{2} + 1\right)$-approximate for $m \geq n$ and $B < nL$, $L < nH$, and $H < nT$.
\end{theorem}
\begin{proof}
 We will proceed by considering all outcomes returned by the mechanism during the ascending phase.

 Let us start by considering the outcome returned at Line~\ref{line:outcome1}. Note that in this case all machines in $A$ have type $B$ and all remaining machines have type larger than $B$. Hence, and by the constraint about types, we have that the outcome of the mechanism is actually optimal.

 Consider now the outcome returned at Line~\ref{line:outcome2} or at Line~\ref{line:outcome4}. Note that in this case we have that there are $\lfloor n/2\rfloor$ machines with type at least $L$, the remaining $\lceil n/2\rceil$ machines with type in $\{B, L\}$, and at least one of these has type $B$. The outcome returned by the mechanism has makespan $mB$. The outcome returned by the optimal mechanism would be to instead evenly split the $m$ jobs over all machines with type $B$: if there are $b$ of these machines then the optimal makespan is $\lceil m/b \rceil B \geq \frac{m}{b} B \geq \frac{m}{\lceil n/2\rceil} B \geq \frac{2m}{n+1} B$. Hence, the approximation ratio is at most $(n+1)/2$. A similar reasoning holds even for the outcome returned at Line~\ref{line:outcome5} simply by replacing $L$ with $H$ and $B$ with $L$.

 Consider now the outcome returned at Line~\ref{line:outcome3}. Note that in this case we have that all machines have type at least $L$ and at least two machines, $i^*$ and $i$, have type exactly $L$. The outcome returned by the mechanism has makespan $\lceil m/2\rceil L$. The outcome returned by the optimal mechanism would be to instead evenly split the $m$ jobs over all machines with type $L$: if there are $\ell$ of these machines then the optimal makespan is $\lceil m/\ell \rceil L \geq \frac{m}{\ell} L \geq \frac{m}{n} L$. Hence the approximation ratio is at most $\frac{\lceil m/2 \rceil}{m/n} \leq \frac{n(m+1)}{2m} = \frac{n}{2} + \frac{n}{2m} \leq \frac{n+1}{2}$. A similar reasoning holds even for the outcome returned at Line~\ref{line:outcome6} and at Line~\ref{line:outcome8} simply by replacing $L$ with $H$ in the first case and with $T$ in the second case.

 Finally, consider the outcome returned at Line~\ref{line:outcome3.5}. Note that $\lfloor n/2\rfloor$ nodes (the non-active ones) have type at least $L$, $i^*$ has type $L$, and the remaining $\lceil n/2\rceil - 1$ machines have type at least $H$. The outcome returned by the mechanism has makespan $mL$. The outcome returned by the optimal mechanism would be to instead evenly split the $m$ jobs over all machines with type $L$: if there are $h$ among non-alive machines then the optimal makespan is $\lceil m/(h+1) \rceil L \geq \frac{m}{h+1} L \geq \frac{m}{\lfloor n/2\rfloor + 1} L \geq \frac{2m}{n+2} L$. Hence, the approximation ratio is at most $(n+2)/2 = \frac{n}{2} + 1$. A similar reasoning holds for the outcome returned at Line~\ref{line:outcome7} simply by replacing $L$ with $H$ and $H$ with $T$.
 \end{proof}

Theorem~\ref{cor:algoapprox} thus proves that the mechanism is tight at least on those instances that have been showed to be hard to approximate in the proof of Theorem~\ref{thm:lb4}. We however note that our mechanism can be easily updated and use routines for the assignment of jobs that are smarter than $\Greedy$. While this may allow to achieve the desired approximation for a larger set of instances, the mechanism and the proof of its OSPness could be significantly more involved. Since the main contribution of this section is to showcase the adoption of our characterization of OSP mechanisms for designing mechanisms and proving their limitations, we here prefer to keep the mechanism and its analysis as simple as possible and leave potential extensions as future work. 

\section{Conclusions}
We have show that OSP and greedy algorithms are intimitely linked, through novel technical (such as, path ironing) and conceptual contributions (e.g., interleaving of greedy approaches as a complex function of outcomes and types). We have applied our results to scheduling related machines and proved that OSP comes at cost that is linear in terms of approximation guarantee already for single-parameter domains of size five. Our result proves that, for scheduling problems, OSP for single dimensional agents is as limiting as strategyproofness for multi-parameter agents \cite{CKK23}. 

Whilst this can be disappointing news, our work builds the framework to better delineate the power of OSP mechanisms. Are there other single-parameter problems where OSP can be closer to strategyproofness (as in the case of single-minded combinatorial auctions \cite{bartmaria})? This question requires understanding better the class of three-way greedy algorithms. For example, how good can they be when they do not interleave? Can the web of interactions between implementation tree, social choice function, and agent's types be disentangled to understand better the power of interleaving?

\newpage
\bibliographystyle{plainnat}
\bibliography{ospb,osps}

\end{document}